\documentclass[10pt,journal]{IEEEtran}
\usepackage{epsfig}
\usepackage{amsmath}
\usepackage{amstext}
\usepackage{amsfonts}
\usepackage{amssymb}
\usepackage{eucal}
\usepackage{graphicx}
\usepackage{verbatim}
\usepackage{stmaryrd}
\usepackage{pstricks}
\usepackage{pst-node}
\usepackage{psfrag}
\usepackage{bm}

\usepackage{algorithmic}

\newcommand{\RR}{{\mathbb{R}}}
\newcommand{\EE}{{\mathrm{E}}}
\newcommand{\NN}{{\mathbb{N}}}

\newcommand{\CC}{{\mathbb{C}}}
\newcommand{\herm}{{\sf H}}
\newcommand{\trans}{{\sf T}}
\renewcommand{\P}{{\bf P}}
\newcommand{\D}{{\bf D}}
\newcommand{\X}{{\bf X}}
\newcommand{\Y}{{\bf Y}}
\newcommand{\W}{{\bf W}}
\newcommand{\A}{{\bf A}}
\newcommand{\uB}{{\underline{\bf B}}}
\newcommand{\Z}{{\bf Z}}
\newcommand{\B}{{\bf B}}
\newcommand{\T}{{\bf T}}

\renewcommand{\H}{{\bf H}}
\newcommand{\I}{{\bf I}}
\renewcommand{\S}{{\bf S}}
\newcommand{\x}{{\bf x}}

\newcommand{\e}{{\bf e}}
\newcommand{\blambda}{{\bm \lambda}}

\newcommand{\y}{{\bf y}}
\renewcommand{\r}{{\bf r}}
\newcommand{\w}{{\bf w}}

\newcommand{\oh}{{\frac{1}{2}}}

\newcommand{\um}{{m_{\uF}}}
\newcommand{\Contour}{{\mathcal C_k}}
\newcommand{\uContour}{{\mathcal C_{G,k}}}
\newcommand{\uuContour}{{\mathcal C_{F,k}}}

\newcommand{\uuGk}{\Gamma_{F,k}}
\newcommand{\uuGkast}{\Gamma^\ast_{F,k}}
\newcommand{\uK}{{{K_G}}}
\newcommand{\uk}{{k_G}}
\newcommand{\uuk}{{k_F}}
\newcommand{\uuj}{{j_F}}

\newcommand{\uY}{{\underline{\Y}}}
\newcommand{\uF}{{\underline{F}}}
\newcommand{\hum}{\hat{m}_{\uF}}

\renewcommand{\hm}{{\hat{m}_F}}
\newcommand{\asto}{\overset{\rm a.s.}{\longrightarrow}}

\DeclareMathOperator{\tr}{tr}
\DeclareMathOperator{\diag}{diag}

\newcounter{ctheorem}
\newtheorem{theorem}[ctheorem]{Theorem}
\newcounter{cproposition}
\newtheorem{proposition}[cproposition]{Proposition}
\newcounter{cassumption}
\newtheorem{assumption}[cassumption]{Assumption}

\newcounter{clemma}
\newtheorem{lemma}[clemma]{Lemma}
\newcounter{cdefinition}
\newtheorem{definition}[cdefinition]{Definition}
\newcounter{cremark}
\newtheorem{remark}[cremark]{Remark}

\begin{document}
\bibliographystyle{IEEEtran}

\title{Eigen-Inference for Energy Estimation\\ of Multiple Sources}

\author{Romain~Couillet,~\IEEEmembership{Student Member,~IEEE,}
  Jack~W.~Silverstein,~Zhidong~Bai,~and~M{\'e}rouane~Debbah,~\IEEEmembership{Senior Member,~IEEE}
  \thanks{R. Couillet and M. Debbah are with the Alcatel-Lucent Chair on Flexible Radio,
    SUP{\'E}LEC, Gif sur Yvette, 91192, Plateau de Moulon, 3, Rue
    Joliot-Curie, France e-mail: \{romain.couillet,~merouane.debbah\}@supelec.fr. M. Debbah's work is supported by the European Commission, FP7 Network of Excellence in Wireless Communications NEWCOM++ and the French ANR Project SESAME.}
  \thanks{J. W. Silverstein is with the Department of Mathematics, North Carolina State University, Raleigh, North Carolina 27695-8205, jack@math.ncsu.edu. J. W. Silverstein's work is supported by the U.S. Army Research Office, Grant W911NF-09-1-0266.}
  \thanks{Z. Bai is with the KLAS MOE \& School of Mathematics and Statistics, Northeast Normal University, Changchun, Jilin 130024, China, baizd@nenu.edu.cn. Z. Bai's work is supported by the NSF China, Grant 10871036, the NUS, Grant R-155-000-079-112 and R-155-000-096-646.}
}

\maketitle

\begin{abstract}
In this paper, a new method is introduced to blindly estimate the transmit power of multiple signal sources in multi-antenna fading channels, when the number of sensing devices and the number of available samples are sufficiently large compared to the number of sources. Recent advances in the field of large dimensional random matrix theory are used that result in a simple and computationally efficient consistent estimator of the power of each source. A criterion to determine the minimum number of sensors and the minimum number of samples required to achieve source separation is then introduced. Simulations are performed that corroborate the theoretical claims and show that the proposed power estimator largely outperforms alternative power inference techniques.
\end{abstract}


\section{Introduction}
At a time when radio resources become scarce, the alternative offered by cognitive radios \cite{MIT99} is gaining more and more interest. A {\it cognitive} (or {\it flexible}) wireless network is a set of opportunistic entities, referred to as the {\it secondary} network, that benefit from unused spectrum resources to establish communication while generating little or even no interference to the licensed networks, collectively referred to as the {\it primary} network. This is achieved by letting the secondary devices sense the communication channel for the presence of active transmissions and exchange the collected information among the secondary network, in order to perform optimal decisions on the opportunistic communication strategy to apply. The difficulty for the secondary network does not lie in the detection of downlink transmissions from fixed access points to licensed mobile users in the primary network, but rather in the reliable detection of the uplink transmissions from the mobile licensed users to the primary access points. If, in addition to detecting active transmissions, the secondary devices can, at all time, detect the exact number of primary mobile sources and evaluate the power used by every individual source, the transmission policy in the secondary network can be accurately and dynamically adapted. An example of use is found in the recent development of femtocells, i.e., small area cells that operate indoors by overlaying the spectrum licensed to outdoors macrocells. Closed access femtocells have the capabilities to self-organize and to dynamically access spectrum resources \cite{CLA08}-\cite{CAL10}; specifically, the first requirement of a femtocell is to minimally interfere the overlaid licensed macrocell network, while simultaneously trying to optimize transmission data rates within the femtocell. This requires that the femtocells be constantly aware of the outdoor activity of the macrocell mobile users. As such, macrocell-femtocell networks are cognitive wireless networks in which the established macrocell network is seen as the primary network, while the femtocell network plays the role of the opportunistic secondary network. In \cite{KON09}, the achievable rates of a two-tier macrocell-femtocell network are derived in the very general case where all entities in the networks are embedded with multiple antennas. The optimal coverage of the secondary networks is computed under several assumptions on the side information available at the femtocells. Among these assumptions, \cite{KON09} supposes that the femtocells have perfect knowledge of the distances to the macrocell user equipments. This last assumption suggests that the femtocells have a means, either global positioning system or some sort of detection mechanism, to perfectly evaluate the distances to the active primary users. In the present work, we address the problem of the estimation of the distance of the secondary network to each primary user or, more exactly, the problem of the estimation of the individual source transmit powers. We provide a framework for the secondary network (i) to identify the number of primary sources, (ii) to determine the number of transmit antennas for every source and (iii) to estimate the transmit power from each source.

The difficulty of estimating transmit powers lies in the little information known {\it a priori} by the secondary network: the transmitted data and the transmission channels are usually inaccessible. This has motivated much work in the direction of blind signal source detection methods, based on the Neyman-Pearson test in Gaussian channels \cite{URK67}, Rayleigh fading channels \cite{KOS02}, multiple antenna channels \cite{COU09b} and large dimensional multi-antenna channels \cite{BIA10}, but these successive works are designed to answer a binary hypothesis test on the presence or absence of a signal source. Alternatively, in \cite{HER07}, a method is derived to separate signal sources and estimate the number of those sources. To solve the harder problem of power inference, it is necessary to assume that the amount of sensors in the secondary network is larger than the number of active sources, e.g., individual secondary users are equipped with many antennas, or a large number of secondary users, each of them equipped with few antennas, collect their received data via a central backbone; this assumption is valid in the context of femtocells that can communicate through wired private or public networks. The condition on the number of sensors allows one to model the multi-dimensional channel $\H$ from the joint primary sources to the secondary users, the joint source transmit data $\X$ and the additive received noise $\W$ as large dimensional random matrices with independent entries (no specific matrix size definition is required at this point). Denoting $\P$ a diagonal matrix whose entries are the source powers with multiplicities the number of transmit antennas of each user, the power detection problem boils down to estimating the entries of $\P$ from the sole knowledge of the received data matrix $\Y=\H\P^\oh\X+\W$, as all system dimensions (number of antennas per transmit source, number of sensors, number of available samples) are large. If the available samples largely outnumber the sensors (of several orders of magnitude), and the number of sensors are much larger than the number of transmit antennas, the strong law of large numbers ensures that the diagonal entries of $\P$ can be retrieved directly from the eigenvalues of $\Y\Y^\herm$, and the problem is immediately solved. When all dimensions are large but are of the same order of magnitude, the law of large numbers no longer applies and one has to consider results from the theory of large dimensional random matrices, e.g., \cite{SIL95}, used in the present article to derive the asymptotic eigenvalue distribution of $\Y\Y^\herm$ as a function of $\P$. To this day and to the best of our knowledge, no computationally-efficient {\it consistent} estimator for the entries of $\P$ has been proposed.\footnote{an estimator $\hat{P}_i$ of the $i^{th}$ entry $P_i$ of $\P$ is said to be {\it consistent} if $\hat{P}_i-P_i\to 0$ almost surely when the relevant system dimensions grow large.} 
Among the existing techniques are discretization and convex optimization strategies \cite{SIL92}, \cite{KAR08}, which tend to directly invert the result from \cite{SIL95} (although an explicit inverse was not available at that time), and moment-based approaches \cite{RAO08}, \cite{COU08}, which use the empirical moments of the eigenvalue distribution $\Y\Y^\herm$ to infer the entries of $\P$. Some of these moment-based methods are computationally cheap, but provide in general consistent estimators of the moments of the eigenvalue distribution of $\P$, instead of estimators of the sought powers. These techniques are therefore expected to perform worse than methods that would fully exploit the eigenvalue distribution of $\Y\Y^\herm$, and not only a few moments of the distribution. This problem is successfully addressed in \cite{MES08} for the simpler {\it sample covariance matrix} model $\Y'=\P^\oh\X$, where strongly consistent estimators for the individual entries of $\P$ are provided, which are based on the full eigenvalue distribution of $\X^\herm\P\X$. 

The present work generalizes this result to infer the entries of $\P$ from the observed matrix $\Y=\H\P^\oh\X+\W$. The novel estimator proposed here is strongly consistent with respect to growing number of sensors, sources and samples, has a very compact form, is computationally efficient and is shown in simulations to largely outperform alternative approaches, such as moment-based methods. The estimator is moreover robust to small system dimensions. We specifically show that, if the number of sensing entities is larger than the number of active transmitters in the primary network, it is possible to evaluate both the exact number of transmitters and their respective transmit powers (and, for that matter, the number of transmit antennas per source can also be estimated). Otherwise, ambiguous scenarios might arise where multiple transmitters may be confused as a single transmitter with estimated transmit power the average of the true transmit powers of these transmitters. Additionally, we provide an expression of the minimum number of sensors required to separate transmit sources of similar power.

The remainder of this paper is structured as follows: Section \ref{sec:model} introduces the system model. In Section \ref{sec:spectral_analysis}, we study the asymptotic spectrum of the eigenvalues of $\Y\Y^\herm$. In Section \ref{sec:estimation}, the novel power estimator is derived. Section \ref{sec:simus} provides simulation results. Section \ref{sec:conclusion} concludes this work.

\textit{Notations:~} In the following, boldface lower case symbols represent vectors, capital boldface characters denote matrices ($\I_N$ is the size-$N$ identity matrix). The transpose and Hermitian transpose operators are denoted $(\cdot)^\trans$ and $(\cdot)^\herm$, respectively. We denote by $\CC^+$ the set $\{z\in \CC, \Im[z]>0\}$ and by $\CC^-$ the set $\{z\in\CC,\Im[z]<0\}$. The left-limit in $x$ of a function $f$ is denoted $f(x-)$.

\begin{figure}
\centering
\includegraphics[width=8cm]{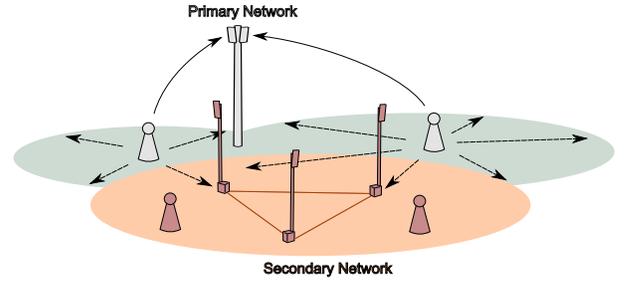}
  \caption{A cognitive radio network}
  \label{fig:scenario}
\end{figure}

\section{System Model}
\label{sec:model}
Consider a wireless (primary) network in which $K$ entities are transmitting data simultaneously on the same frequency resource. Transmitter $k\in\{1,\ldots,K\}$ has transmission power $P_k$ and is equipped with $n_k$ antennas. We denote $n=\sum_{k=1}^K n_k$ the total number of transmit antennas of the primary network. Consider also a secondary network composed of a total of $N$, $N\geq n$, sensing devices (they may be $N$ single antenna devices or multiple devices embedded with multiple antennas whose sum equals $N$); we shall refer to the $N$ sensors collectively as {\it the receiver}. This scenario is depicted in Figure \ref{fig:scenario}. To ensure that every sensor in the secondary network, e.g., in a femtocell, roughly captures the same amount of energy from a given transmitter, we need to assume that the respective transmitter-sensor distances are alike. This is a realistic assumption for a in-house femtocell network. Denote $\H_k\in \CC^{N\times n_k}$ the multiple antenna channel matrix between transmitter $k$ and the receiver. We assume that the entries of $\sqrt{N}\H_k$ are independent and identically distributed with zero mean, unit variance and finite fourth order moment. At time instant $m$, transmitter $k$ emits the multi-antenna signal $\x^{(m)}_k\in\CC^{n_k}$, with entries assumed to be independent and identically distributed of zero mean, unit variance and finite fourth order moment. Assume further that at time instant $m$ the receive signal is impaired by additive white noise with entries of zero mean, variance $\sigma^2$ and finite fourth order moment on every sensor; we denote $\sigma\w^{(m)}\in \CC^{N}$ the receive noise vector where the entries of $\w^{(m)}_k$ have unit variance. At time $m$, the receiver therefore senses the signal $\y^{(m)}\in \CC^N$ defined as
\begin{equation}
	\y^{(m)} = \sum_{k=1}^K \sqrt{P_k}\H_k\x^{(m)}_k + \sigma\w^{(m)}.
\end{equation}
Assuming the channel fading coefficients are constant over at least $M$ consecutive sampling periods, by concatenating $M$ successive signal realizations into $\Y=[\y^{(1)},\ldots,\y^{(M)}]\in \CC^{N\times M}$, we have
\begin{equation}
	\Y = \sum_{k=1}^K \sqrt{P_k}\H_k\X_k + \sigma\W
\end{equation} 
where $\X_k=[\x^{(1)}_k,\ldots,\x^{(M)}_k]\in \CC^{n_k\times M}$, for every $k$, and $\W=[\w^{(1)},\ldots,\w^{(M)}]\in \CC^{N\times M}$. This can be further rewritten as
\begin{equation}
	\label{eq:HPX+W}
	\Y = \H\P^\oh\X + \sigma\W
\end{equation}
where $\P\in \RR^{n\times n}$ is diagonal with first $n_1$ entries $P_1$, subsequent $n_2$ entries $P_2$, \ldots, last $n_K$ entries $P_K$, $\H=[\H_1,\ldots,\H_K]\in \CC^{N\times n}$ and $\X=[\X_1^\trans,\ldots,\X_K^\trans]^\trans\in \CC^{n\times M}$. By convention, we shall assume $P_1\leq \ldots \leq P_K$.

\begin{remark}
  The statement that $\sqrt{N}\H$, $\X$ and $\W$ have independent entries of finite fourth order moment is meant to provide as loose assumptions as possible on the channel, signal and noise properties. In the simulations of Section \ref{sec:simus}, the entries of $\H$, $\W$ are taken Gaussian. Nonetheless, according to our assumptions, the entries of $\X$ need not be identically distributed, but may originate from a maximum of $K$ distinct distributions. This translates the realistic assumption that different data sources may use different symbol constellations (e.g., $M$-QAM, $M$-PSK); the finite fourth moment assumption is obviously verified for finite constellations.
\end{remark}

Our objective is to infer the values of the powers $P_1,\ldots,P_K$ from the realization of the random matrix $\Y$. This is the subject of Section \ref{sec:estimation}. In the sequel, we introduce tools from large dimensional random matrix theory and we provide a thorough analysis of the eigenvalue distribution of $\frac1M\Y\Y^\herm$ as $N$, $n$ and $M$ grow large at the same rate.

\section{Spectral analysis}
\label{sec:spectral_analysis}
We start by analyzing the eigenvalue distribution of $\frac1M\Y\Y^\herm$ when $n$, $N$ and $M$ grow large at a similar rate. This is a fundamental prior step to the proper estimation of $P_1,\ldots,P_K$.

\subsection{Limiting spectrum of $\frac1M\Y\Y^\herm$}
We first define the Stieltjes transform of a (cumulative) distribution function.

\begin{definition}
	Let $F$ be a distribution function. For $z\in\CC\setminus \RR^+$, the {\it Stieltjes transform} $m(z)$ of $F$ is defined as
\begin{equation}
	\label{eq:stieltjes}
	m(z) = \int \frac1{t-z}dF(t).
\end{equation}

For $x\in\RR$ a continuity point of $F$, we have the {\it inverse Stieltjes transform} formula
\begin{equation}
	\label{eq:inv_stieltjes}
	F(x) = \frac1{\pi}\lim_{y\to 0^+} \int_{-\infty}^x \Im[m(t+iy)]dt.
\end{equation}
\end{definition}

In this section, we prove the following result
\begin{theorem}
	\label{th:1}
	Let $\B_N=\frac1M\Y\Y^\herm$, with $\Y$ defined as in \eqref{eq:HPX+W}. Then, for $M$, $N$, 
$n$ growing large with limit ratios $M/N\to c$, $N/n_k\to c_k$, $0<c,c_1,\ldots,c_K<\infty$, the 
eigenvalue distribution function $F^{\B_N}$ of $\B_N$, referred to as the {\it empirical spectral 
distribution} (e.s.d.) of $\B_N$, converges almost surely to the deterministic distribution function 
$F$, referred to as the {\it limit spectral distribution} (l.s.d.) of $\B_N$, whose Stieltjes 
transform $m_F(z)$ satisfies, for $z\in\CC^+$,
\begin{equation}
	\label{eq:m_of_um}
	m_F(z) = c\um(z)+(c-1)\frac1z
\end{equation}
where $\um(z)$ is the unique solution with positive imaginary part of the implicit equation in $\um$
\begin{equation}
	\label{eq:um_fp}
	\frac1{\um} = -\sigma^2+\frac1{f}-\sum_{k=1}^K\frac1{c_k}\frac{P_k}{1+P_k f}
\end{equation}
in which we denoted $f$ the value
\begin{equation}
	f = (1-c) \um - c z\um^2.
\end{equation}
\end{theorem}

The rest of this section is dedicated to the proof of Theorem \ref{th:1}. First remark that \eqref{eq:HPX+W} can be further simplified into
\begin{equation}
	\Y = \begin{pmatrix} \H\P^\oh & \sigma \I_N \end{pmatrix}\begin{pmatrix} \X \\ \W \end{pmatrix}.
\end{equation}

Appending $\Y\in\CC^{N\times M}$ into the larger matrix $\uY\in \CC^{(N+n)\times M}$
\begin{equation}
	\uY = \begin{pmatrix} \H\P^\oh & \sigma \I_N \\ 0 & 0 \end{pmatrix}\begin{pmatrix} \X \\ \W \end{pmatrix},
\end{equation}
we recognize that $\frac1M\uY\uY^\herm$ is a {\it sample covariance matrix}, for which the {\it population covariance matrix} $\left(\begin{smallmatrix} \H\P\H^\herm + \sigma^2\I_N & 0 \\ 0 & 0 \end{smallmatrix}\right)$ is non-deterministic and the random matrix $\left(\begin{smallmatrix} \X \\ \W \end{smallmatrix}\right)$ has independent (non-necessarily identically distributed) entries with zero mean and variance $1$.
  
  At this point, we need the following result,
  \begin{proposition}
    \label{prop:sil95}
    Let $\Z_n\in \CC^{N\times n}$ have complex independent entries of zero mean, unit variance and finite $2+\varepsilon$ order moment, for some $\varepsilon>0$, and $\T_n\in \RR^{n\times n}$ be Hermitian with e.s.d. converging almost surely to $T$, as $N\to \infty$. Let $\A_n=\frac1N\Z_n\T_n\Z_n^\herm$. Then, as $n,N\to \infty$, $N/n\to c>0$, the eigenvalue distribution of $\A_n$ converges weakly and almost surely to the distribution function $A$ with Stieltjes transform $m_A(z)$, $z\in\CC^+$, being the unique solution with positive imaginary part of the equation in $m_A$
    \begin{equation}
      \label{eq:sil95}
      z = -\frac1{m_A} + \frac1c \int \frac{t}{1+t m_A}dT(t).
    \end{equation}
  \end{proposition} 
  \begin{proof}
    The proof originates from Theorem 4.1 of \cite{SIL06} that states that, under the hypotheses of Proposition \ref{prop:sil95}, the eigenvalue distribution of $\A_n$ converges weakly to some distribution function $A$ whose Stieltjes transform $m_A(z)$ is a function of the Stieltjes transform of $m_T(z)$ and $c$ only; $m_A(z)$ is explicitly given by (4.4.4) of \cite{SIL06}. Now, in the special case where $\Z_n$ has independent and identically distributed (i.i.d.) entries of zero mean, unit variance and finite $2+\varepsilon$ moment, \cite{SIL95} and Theorem 4.3 of \cite{SIL06} show that $m_A(z)$ satisfies \eqref{eq:sil95}. But then, since $m_A(z)$ is only a function of $c$ and $T$ regardless of the distribution of the independent entries of $\Z_n$, $m_A(z)$ that solves \eqref{eq:sil95} is the Stieltjes transform of $A$ for the more general case.
\end{proof}

 Note that Proposition \ref{prop:sil95} can be equally stated when $z\in\CC^-$. In that case, $m_A(z)$ is the unique solution of \eqref{eq:sil95} with negative imaginary part.

  From Proposition \ref{prop:sil95}, since $\H$ has independent entries with finite fourth order moment, we have that the e.s.d. of $\H\P\H^\herm$ converges weakly and almost surely to a limit distribution $G$ as $N,n_1,\ldots,n_K\to \infty$ with, $N/n_k\to c_k>0$. For $z\in \CC^+$, the Stieltjes transform $m_G(z)$ of $G$ is the unique solution with positive imaginary part of the equation in $m_G$,
  \begin{equation}
    \label{eq:m1}
    z = -\frac1{m_G} + \sum_{k=1}^K \frac1{c_k}\frac{P_k}{1+P_km_G}.
  \end{equation}
  
  The almost sure convergence of the e.s.d. of $\H\P\H^\herm$ ensures the almost sure convergence of the e.s.d. of the matrix $\left(\begin{smallmatrix} \H\P\H^\herm + \sigma^2\I_N & 0 \\ 0 & 0 \end{smallmatrix}\right)$. Since $m_G(z)$ evaluated at $z\in\CC^+$ is the Stieltjes transform of the l.s.d. of $\H\P\H^\herm+\sigma^2\I_N$ evaluated at $z+\sigma^2$, adding $n$ zero eigenvalues, we finally have that the e.s.d. of $\left(\begin{smallmatrix} \H\P\H^\herm + \sigma^2\I_N & 0 \\ 0 & 0 \end{smallmatrix}\right)$ tends almost surely to a distribution $H$ whose Stieltjes transform $m_H(z)$ satisfies
    \begin{equation}
      \label{eq:mH}
      m_H(z) = \frac{c_0}{1+c_0}m_G(z-\sigma^2) - \frac1{1+c_0}\frac1z,
    \end{equation}
for $z\in\CC^+$, where we denoted $c_0$ the limit of the ratio $N/n$, i.e., $c_0=(c_1^{-1}+\ldots+c_K^{-1})^{-1}$.  

As a consequence, the sample covariance matrix $\frac1M\uY\uY^\herm$ has a population covariance matrix which is not deterministic but whose e.s.d. has an almost sure limit for increasing dimensions. Since $\X$ and $\W$ have entries with finite fourth order moment, we can again apply Proposition \ref{prop:sil95}, and we have that the e.s.d. of $\uB_N\triangleq \frac1M\uY^\herm\uY$ converges almost surely to the limit $\uF$ whose Stieltjes transform $\um(z)$ is the unique solution in $\CC^+$ of the equation in $\um$

\begin{align}
	\label{eq:z_of_m}
	z &= -\frac1{\um} + \frac1c\left(1+\frac1{c_0}\right) \int \frac{t}{1+t\um}dH(t) \\
	\label{eq:z_of_m2}
	&= -\frac1{\um} + \frac{1+\frac1{c_0}}{c\um}\left[1 - \frac1{\um}m_H\left(-\frac1{\um}\right)\right]
\end{align}
for all $z\in\CC^+$.
	
For $z\in\CC^+$, $\um(z)\in\CC^+$. Therefore $-1/\um(z)\in \CC^+$ and one can evaluate \eqref{eq:mH} at $-1/\um(z)$. Combining \eqref{eq:mH} and \eqref{eq:z_of_m2}, we then have
	\begin{equation}
		\label{eq:z_of_m1_of_m}
		z = -\frac1c \frac1{\um(z)^2}m_G\left(-\frac1{\um(z)}-\sigma^2\right)+\left(\frac1c-1\right)\frac1{\um(z)},
	\end{equation}
where, according to \eqref{eq:m1}, $m_G(-1/\um(z)-\sigma^2)$ satisfies
\begin{align}
  \frac1{\um(z)} = & -\sigma^2+\frac1{m_G(-\frac1{\um(z)}-\sigma^2)} \nonumber \\ & -\sum_{k=1}^K\frac1{c_k}\frac{P_k}{1+P_km_G(-\frac1{\um(z)}-\sigma^2)}.
\end{align}

Together with \eqref{eq:z_of_m1_of_m}, this is exactly \eqref{eq:um_fp}, with $f(z)=m_G(-\frac1{\um(z)}-\sigma^2)=(1-c)m_{\uF}(z)-czm_{\uF}(z)^2$.

Since the eigenvalues of the matrices $\B_N$ and $\uB_N$ only differ by $M-N$ zeros, we also have that the Stieltjes transform $m_F(z)$ of the l.s.d. of $\B_N$ satisfies
\begin{equation}
  \label{eq:m_um}
	m_F(z) = c\um(z)+(c-1)\frac1z.
\end{equation}

This completes the proof of Theorem \ref{th:1}. For further usage, notice here that \eqref{eq:m_um} provides a simplified expression for $m_G(-1/\um(z)-\sigma^2)$. Indeed we have,
\begin{equation}
  \label{eq:m1_m_um}
m_G(-1/\um(z)-\sigma^2) = -z m_F(z)\um(z).
\end{equation}

Therefore, the support of the (almost sure) l.s.d. $F$ of $\B_N$ can be evaluated as follows: for any $z\in\CC^+$, $m_F(z)$ is given by \eqref{eq:m_of_um}, in which $\um(z)$ is solution of \eqref{eq:um_fp}; the inverse Stieltjes transform formula \eqref{eq:inv_stieltjes} allows then to evaluate $F$ from $m_F(z)$, for values of $z$ spanning over the set $\{z=x+iy,x>0\}$ and $y$ small. This is depicted in Figure \ref{fig:spectrum}, when $\P$ has three distinct values $P_1=1$, $P_2=3$, $P_3=10$ and $n_1=n_2=n_3$, $N/n=10$, $M/N=10$, $\sigma^2=0.1$, as well as in Figure \ref{fig:spectrum2} for the same setup but with $P_3=5$.

\begin{figure}
  \centering
\includegraphics[]{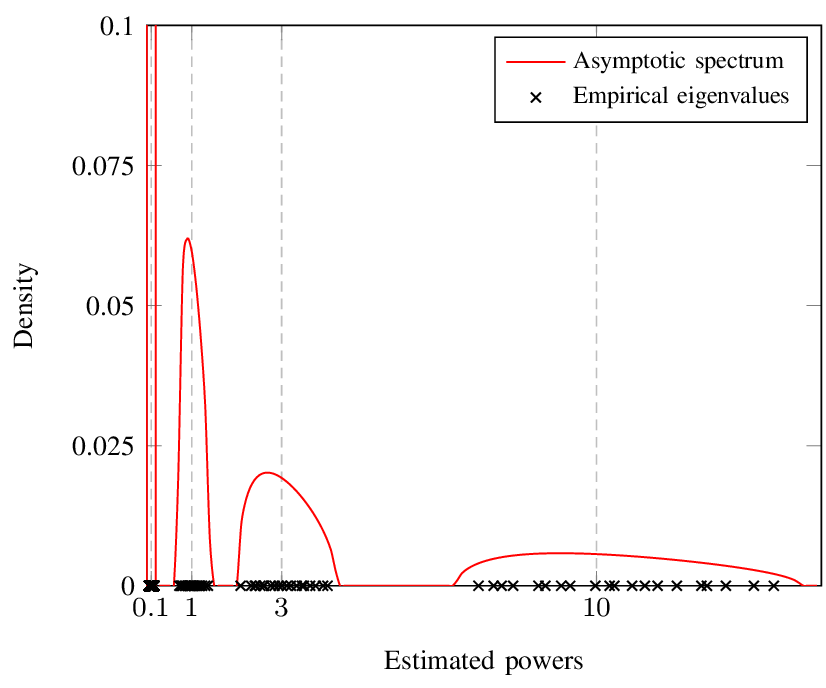}
  \caption{Empirical and asymptotic eigenvalue distribution of $\frac1M\Y\Y^\herm$ when $\P$ has three distinct entries $P_1=1$, $P_2=3$, $P_3=10$, $n_1=n_2=n_3$, $c_0=10$, $c=10$, $\sigma^2=0.1$. Empirical test: $n=60$.}
  \label{fig:spectrum}
\end{figure}

\begin{figure}
  \centering
\includegraphics[]{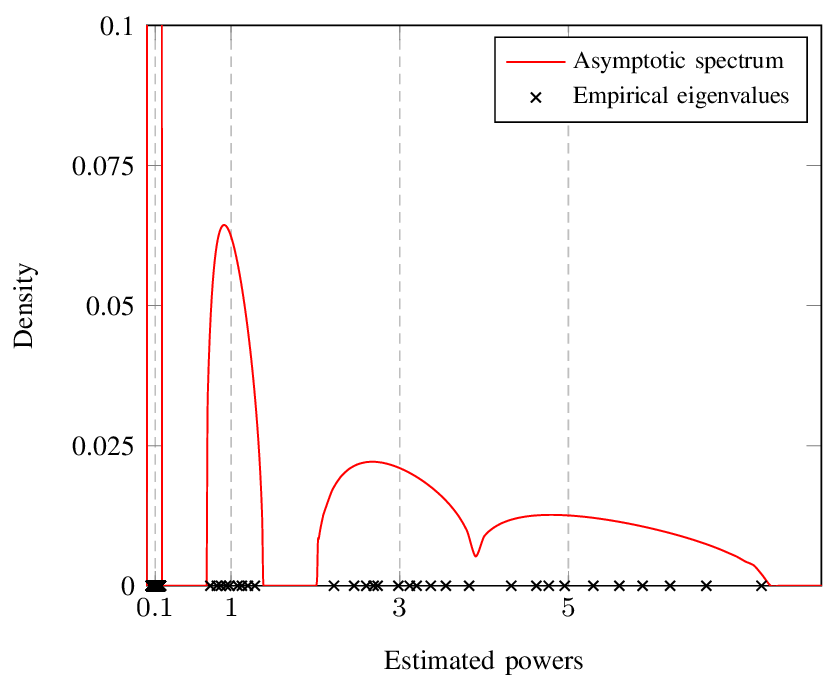}
  \caption{Empirical and asymptotic eigenvalue distribution of $\frac1M\Y\Y^\herm$ when $\P$ has three distinct entries $P_1=1$, $P_2=3$, $P_3=5$, $n_1=n_2=n_3$, $c_0=10$, $c=10$, $\sigma^2=0.1$. Empirical test: $n=60$.}
  \label{fig:spectrum2}
\end{figure}

Two remarks on Figures \ref{fig:spectrum} and \ref{fig:spectrum2} are of fundamental importance to the following. First, it appears that the asymptotic l.s.d. $F$ of $\B_N$ is compactly supported and divided into up to $K+1$ disjoint compact intervals, which we further refer to as {\it clusters}. Each cluster can be mapped onto one or many values in the set $\{\sigma^2,P_1,\ldots,P_K\}$. For instance, in Figure \ref{fig:spectrum2}, the first cluster is mapped to $\sigma^2$, the second cluster to $P_1$ and the third cluster to the set $\{P_2,P_3\}$. Depending on the ratios $c$ and $c_0$ and on the particular values taken by $P_1,\ldots,P_K$ and $\sigma^2$, these clusters are either thin disjoint compact intervals, as in Figure \ref{fig:spectrum}, or they may overlap to generate larger compact intervals, as in Figure \ref{fig:spectrum2}. We shall see, as is in fact required by the law of large numbers, that for increasing $c$ and $c_0$, the asymptotic spectrum tends to be divided into thinner and thinner clusters. The inference technique proposed hereafter relies on the separability of the clusters associated to each $P_i$ and to $\sigma^2$. Precisely, to be able to derive a consistent estimate of the transmitted power $P_k$, the cluster associated to $P_k$ in $F$, number it cluster $k_F$, must be distinct from the neighboring clusters $(k-1)_F$ and $(k+1)_F$, associated to $P_{k-1}$ and $P_{k+1}$ respectively (when they exist), and also distinct from cluster $1$ in $F$ associated to $\sigma^2$. As such, in the scenario of Figure \ref{fig:spectrum2}, our method will be able to provide a consistent estimate for $P_1$, but (so far) will not succeed in providing a consistent estimate for either $P_2$ or $P_3$, since $2_F=3_F$. We shall see that a consistent estimate for $(P_2+P_3)/2$ is accessible though. Secondly, notice that the empirical eigenvalues of $\B_N$ are all inside the asymptotic clusters and, most importantly, in the case where cluster $k_F$ is distinct from $1$, $(k-1)_F$ and $(k+1)_F$, observe that the number of eigenvalues in cluster $k_F$ is exactly $n_k$. This fact is referred to as {\it exact separation}. The exact separation for the current model originates from a direct application of the exact separation for the sample covariance matrix proven in \cite{BAI99} and is provided here in Theorem \ref{th:bai}. This is further discussed in the subsequent sections.

\subsection{Condition for separability}
In the following, we are interested in estimating consistently the power $P_k$ for a given fixed $k\in\{1,\ldots,K\}$. We recall that consistency means here that, as all system dimensions grow large with finite asymptotic ratios, the difference $\hat{P}_k-P_k$ between the estimate $\hat{P}_k$ of $P_k$ and $P_k$ itself converges to zero with probability one. As previously mentioned, we show by construction in Section \ref{sec:estimation} that such an estimate is only achievable if the cluster mapped to $P_k$ in $F$ is disjoint from all other clusters. The purpose of the present section is to provide sufficient conditions for cluster separability. 

To ensure that cluster $k_F$ (associated to $P_k$ in $F$) is distinct from cluster $1$ (associated to $\sigma^2$) and clusters $i_F$, $i\neq k$, (associated to all other $P_i$), we assume now and for the rest of this article that the following conditions are fulfilled: (i) $k$ satisfies Assumption \ref{ass:1}, given as follows
\begin{assumption}
  \label{ass:1}
  \begin{align}
	  \label{eq:ass1}
	  \sum_{r=1}^K \frac1{c_r}\frac{\left(P_r m_{G,k}\right)^2}{\left(1+ P_r m_{G,k} \right)^2}&< 1 \\
	  \sum_{r=1}^K \frac1{c_r}\frac{\left(P_r m_{G,k+1}\right)^2}{\left(1+ P_r m_{G,k+1} \right)^2}&< 1
  \end{align}
  with $m_{G,1},\ldots,m_{G,K}$ the $K$ real solutions to the equation in $m_G$,
  \begin{equation}
    \label{eq:ass12}
    \sum_{r=1}^K \frac1{c_r}\frac{\left(P_r m_G\right)^3}{\left(1+ P_r m_G \right)^3} = 1
  \end{equation}
  with the convention $m_{G,K+1}=0$,
\end{assumption}
and (ii) $k$ satisfies Assumption \ref{ass:2} as follows,
\begin{assumption}
  \label{ass:2}
  Denoting, for $j\in\{1,\ldots,K\}$, 
  \begin{equation}
j_G \triangleq \# \left\{i\leq j ~\vert~ i \textmd{ satisfies Assumption \ref{ass:1}}\right\},
\end{equation}
  \begin{align}
	  \label{eq:ass2}
& \frac{1-c_0}{c_0}\frac{(\sigma^2m_{\uF,\uk})^2}{(1+\sigma^2m_{\uF,\uk})^2} \nonumber \\
&+ \sum_{r=1}^{\uk-1} \frac1{c_r}\frac{(x_{G,r}^++\sigma^2)^2m_{\uF,\uk}^2}{(1+(x_{G,r}^++\sigma^2)m_{\uF,\uk})^2} \nonumber \\
&+ \sum_{r=\uk}^\uK \frac1{c_r}\frac{(x_{G,r}^-+\sigma^2)^2m_{\uF,\uk}^2}{(1+(x_{G,r}^-+\sigma^2)m_{\uF,\uk})^2} < c \\
&\frac{1-c_0}{c_0}\frac{(\sigma^2m_{\uF,\uk+1})^2}{(1+\sigma^2m_{\uF,\uk+1})^2} \nonumber \\
&+ \sum_{r=1}^{\uk} \frac1{c_r}\frac{(x_{G,r}^++\sigma^2)^2m_{\uF,\uk+1}^2}{(1+(x_{G,r}^++\sigma^2)m_{\uF,\uk+1})^2} \nonumber \\
&+ \sum_{r=\uk+1}^\uK \frac1{c_r}\frac{(x_{G,r}^-+\sigma^2)^2m_{\uF,\uk+1}^2}{(1+(x_{G,r}^-+\sigma^2)m_{\uF,\uk+1})^2} < c
  \end{align}
  where $x_{G,i}^-,x_{G,i}^+$, $i\in\{1,\ldots,\uK\}$, are defined by 
  \begin{align}
    x_{G,i}^- &= -\frac1{m_{G,i}^-} + \sum_{r=1}^K \frac1{c_r}\frac{P_r}{1+P_rm_{G,i}^-} \\
    x_{G,i}^+ &= -\frac1{m_{G,i}^+} + \sum_{r=1}^K \frac1{c_r}\frac{P_r}{1+P_rm_{G,i}^+} 
  \end{align}
  with $m_{G,1}^-,m_{G,1}^+,\ldots,m_{G,\uK}^-,m_{G,\uK}^+$ the $2\uK$ real roots of \eqref{eq:ass1}, and $m_{\uF,j}$, $j\in\{1,\ldots,\uK+1\}$, the $j$-th real root (in increasing order) of the equation in $\um$
  \begin{align}
\frac{1-c_0}{c_0}\frac{(\sigma^2\um)^3}{(1+\sigma^2\um)^3} &+ \sum_{r=1}^{j-1} \frac1{c_r}\frac{(x_{G,r}^++\sigma^2)^3\um^3}{(1+(x_{G,r}^++\sigma^2)\um)^3} \nonumber \\
&+ \sum_{r=j}^\uK \frac1{c_r}\frac{(x_{G,r}^-+\sigma^2)^3\um^3}{(1+(x_{G,r}^-+\sigma^2)\um)^3} = c. 
  \end{align}
\end{assumption}

Although difficult to fathom at this point of the article, the above assumptions will be clarified in the subsequent sections. We give hereafter a short intuitive explanation of the role of every condition. 


Assumption \ref{ass:1} is a necessary and sufficient condition for cluster $k_G$, that we define as the cluster associated to $P_k$ in $G$ (the l.s.d. of $\H\P\H^\herm$), to be distinct from the clusters $(k-1)_G$ and $(k+1)_G$, associated to $P_{k-1}$ and $P_{k+1}$ in $G$, respectively. Note that we implicitly assume a unique mapping between the $P_i$ and clusters in $G$; this statement will be made more rigorous in subsequent sections. Assumption \ref{ass:1} only deals with the inner $\H\P\H^\herm$ covariance matrix properties and ensures specifically that the powers to be estimated differ sufficiently from one another for our method to be able to resolve them. Note that, if $P_1,\ldots,P_K$ are scaled by a common constant, then the solutions of \eqref{eq:ass12} are scaled by the inverse of this constant; the separability condition is then a function of $P_2/P_1,\ldots,P_K/P_1$ and of the ratios $c_1,\ldots,c_K$ only. In Figure \ref{fig:ass1_c0}, we depict the critical ratio $c_0$ above which Assumption \ref{ass:1} is satisfied for all $k$, when $K=2$ and $c_1=c_2$, as a function of $P_1/P_2$, i.e., the critical ratio $c_0$ above which the two clusters associated to $P_1$ and $P_2$ in $G$ are disjoint. Observe that, as $P_1$ gets close to $P_2$, $c_0$ increases fast; therefore, to be able to separate power values with ratio close to one, an extremely large number of sensors is required. In Figure \ref{fig:ass1_3}, the case $K=3$ is considered with $c_1=c_2=c_3$, $c_0=10$, and we let $P_2/P_1$ and $P_3/P_1$ vary; this situation corresponds to the scenarios previously depicted in Figures \ref{fig:spectrum} and \ref{fig:spectrum2}. Note that the triplet $(P_1,P_2,P_3)=(1,3,5)$ is slightly outside the region that satisfies Assumption \ref{ass:1}, and then, for this $c_0$, not all the clusters of $G$ (and therefore of $F$) are disjoint, as confirmed by Figure \ref{fig:spectrum2}. As for the triplet $(1,3,10)$, it clearly lies inside the region that satisfies Assumption \ref{ass:1}, which is sufficient to ensure the separability of the clusters in $G$, but not enough though to ensure the separability of the clusters in $F$.

Assumption \ref{ass:2} deals with the complete $\B_N$ matrix model. It is however a non-necessary but sufficient condition so that cluster $k_F$, associated to $P_k$ in $F$, be distinct from clusters $(k-1)_F$, $(k+1)_F$ and $1$ (cluster $1$ being associated to $\sigma^2$). The exact necessary and sufficient condition will be stated further in the next sections; however, the latter is not exploitable as is, and Assumption \ref{ass:2} will be shown to be an appropriate substitute. Assumption \ref{ass:2} is concerned with the value of $c$ necessary to avoid (i) cluster $k_G$ (associated to $P_k$ in $G$) to further spread on the clusters $\uk-1$ and $\uk+1$ associated to $P_{k-1}$ and $P_{k+1}$ and, more importantly, to avoid (ii) cluster $1$ associated to $\sigma^2$ in $F$ to merge with cluster $\uuk$. As shall become evident in the next sections, when $\sigma^2$ is large, the tendency is for the cluster associated to $\sigma^2$ to become large and spread over the clusters associated to $P_1$, then $P_2$ etc. To counter this effect, one must increase $c$, i.e., take more signal samples. Figure \ref{fig:ass2} depicts the critical ratio $c$ that satisfies Assumption \ref{ass:2} as a function of $\sigma^2$, in the case $K=3$, $(P_1,P_2,P_3)=(1,3,10)$, $c_0=10$, $c_1=c_2=c_3$. Notice that, in the case $c=10$, below $\sigma^2\simeq 1$, it is possible to separate all clusters, which is compliant with Figure \ref{fig:spectrum} where $\sigma^2=0.1$.

As a consequence, under the assumption (proved later) that our proposed method cannot perform consistent power estimation when the cluster separability conditions are not met, we have two first conclusions:
\begin{itemize}
  \item if one desires to increase the sensitivity of the estimator, i.e., to be able to separate two sources of close transmit powers, one needs to increase the number of sensors (by increasing $c_0$),
  \item if one desires to detect and reliably estimate power sources in a noisy environment, one needs to increase the number of sensed samples (by increasing $c$).
\end{itemize}

In the subsequent section, we study the properties of the asymptotic spectrum of $\H\P\H^\herm$ and $\uB_N$ in more details. These properties will lead to an explanation for Assumptions \ref{ass:1} and \ref{ass:2}. Under those assumptions, we shall then derive our novel power estimator.

\begin{figure}
\centering
\includegraphics[]{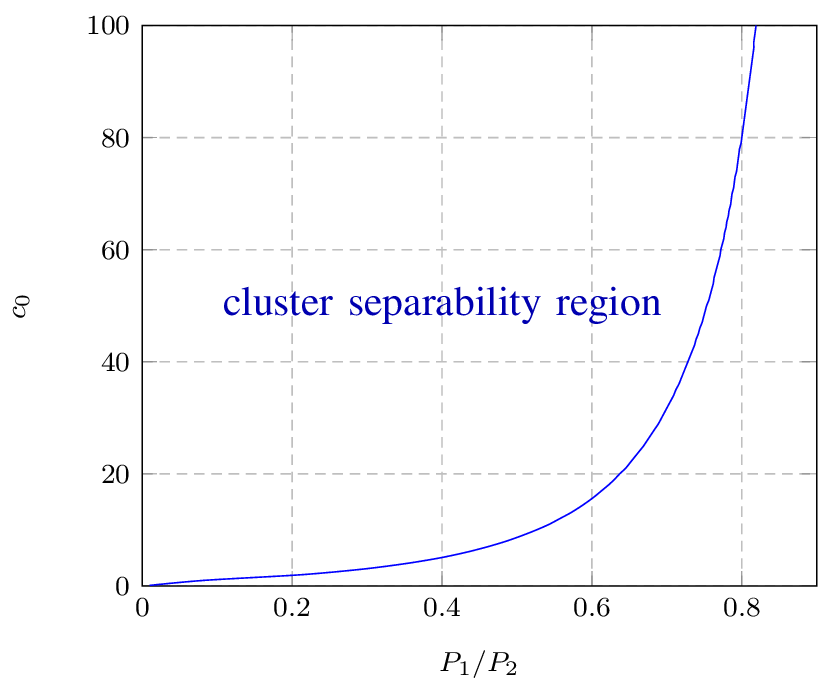}
  \caption{Limiting ratio $c_0$ to ensure separability of $(P_1,P_2)$, $P_1\leq P_2$, $K=2$, $c_1=c_2$.}
  \label{fig:ass1_c0}
\end{figure}

\begin{figure}
\centering
\includegraphics[]{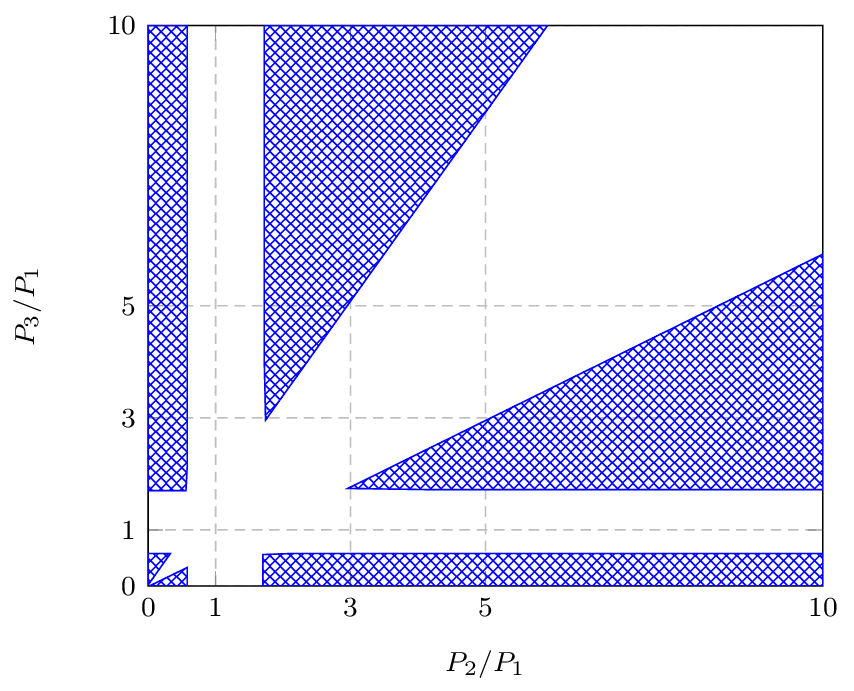}
  \caption{Subset of $(P_1,P_2,P_3)$ that fulfills Assumption \ref{ass:1} $K=3$, $c_1=c_2=c_3$, for $c_0=10$, in crosshatched pattern.}
  \label{fig:ass1_3}
\end{figure}

\begin{figure}
\centering
\includegraphics[]{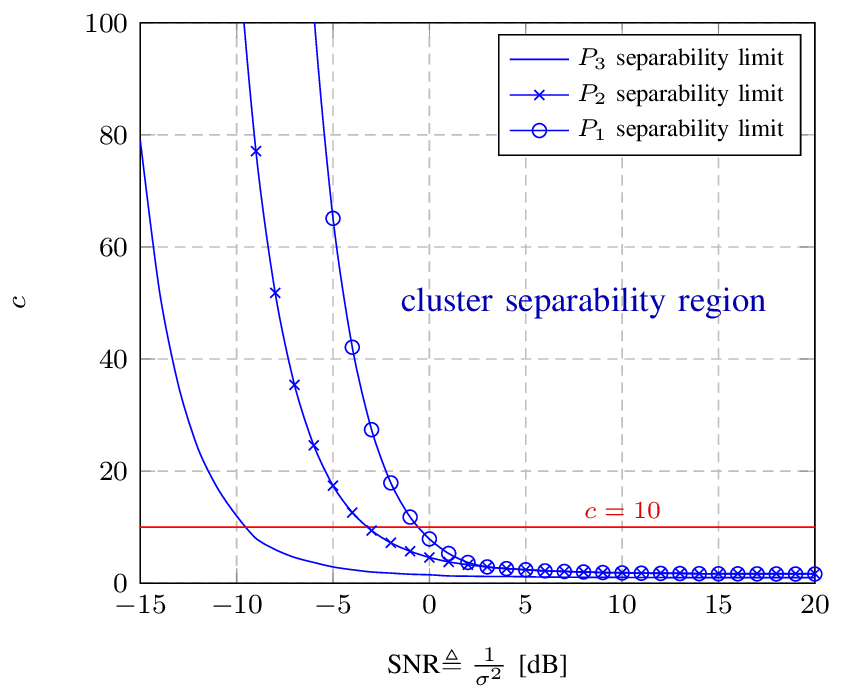}
  \caption{Limiting ratio $c$ as a function of $\sigma^2$ to ensure consistent estimation of $P_1=1$, $P_2=3$ and $P_3=10$, $c_0=10$, $c_1=c_2=c_3$.}
  \label{fig:ass2}
\end{figure}

\section{Multi-source power inference}
\label{sec:estimation}
In this section, we prove our main result,
\begin{theorem}
	\label{th:2}
	Let $\B_N\in\CC^{N\times N}$ be defined as in Theorem \ref{th:1}, and ${\blambda}=(\lambda_1,\ldots,\lambda_N)$, $\lambda_1\leq \ldots\leq \lambda_N$, be the vector of the ordered eigenvalues of $\B_N$. Further assume that the limiting ratios $c_0$, $c_1,\ldots,c_K$, $c$ and $\P$ are such that Assumptions \ref{ass:1} and \ref{ass:2} are fulfilled for some $k\in\{1,\ldots,K\}$. Then, as $N$, $n$, $M$ grow large, we have
	\begin{equation}
		\label{eq:limPk}
		\hat{P}_k - P_k \asto 0
	\end{equation}
	where the estimate $\hat{P}_k$ is given by
	\begin{itemize}
	  \item if $M\neq N$,
	    \begin{equation}
	      \label{eq:th_Pk}
	      \hat{P}_k = \frac{NM}{n_k(M-N)}\sum_{i\in \mathcal N_k} (\eta_i-\mu_i),
	    \end{equation}
	  \item if $M=N$,
	    \begin{equation}
	      \label{eq:th_Pk2}
	      \hat{P}_k = \frac{N}{n_k(N-n)}\sum_{i\in \mathcal N_k} \left(\sum_{j=1}^N \frac{\eta_i}{(\lambda_j-\eta_i)^2}\right)^{-1},
	    \end{equation}
      \end{itemize}
      in which $\mathcal N_k=\{\sum_{i=1}^{k-1} n_i +1,\ldots,\sum_{i=1}^k n_i\}$, $(\eta_1,\ldots,\eta_N)$ are the ordered eigenvalues of the matrix $\diag(\blambda)-\frac1N\sqrt{\blambda}\sqrt{\blambda}^\trans$ and $(\mu_1,\ldots,\mu_N)$ are the ordered eigenvalues of the matrix $\diag(\blambda)-\frac1M\sqrt{\blambda}\sqrt{\blambda}^\trans$.
\end{theorem}

\begin{remark}
  \label{rem:c0}
  We immediately notice that, if $N<n$, the powers $P_1,\ldots,P_l$, with $l$ the largest integer such that $N-\sum_{i=l}^K n_i<0$, cannot be estimated.
\end{remark}

The approach pursued to prove Theorem \ref{th:2} relies strongly on the original idea of \cite{MES08}. From Cauchy's integration formula \cite{RUD86},
\begin{align}
  P_k &= \frac1{2\pi i}\oint_{\mathcal C_k} \frac{\omega}{P_k-\omega}d\omega \nonumber \\
	\label{eq:Pkint}
	 &= c_k\frac1{2\pi i}\oint_{\mathcal C_k}\sum_{r=1}^{K}\frac1{c_r}\frac{\omega}{P_r-\omega}d\omega
\end{align}
for any negatively oriented contour $\mathcal C_k\subset \CC$, such that $P_k$ is contained in the surface described by the contour, while for every $i\neq k$, $P_i$ is outside this surface. The strategy is then the following: we first propose a convenient integration contour $\mathcal C_k$ which is parametrized by a function of the Stieltjes transform $m_F(z)$ of the l.s.d. of $\B_N$. We proceed to a variable change in \eqref{eq:Pkint} to express $P_k$ as a function of $m_F(z)$. We then evaluate the complex integral resulting from replacing the limiting $m_F(z)$ in \eqref{eq:Pkint} by its empirical counterpart $\hat{m}_F(z)=\frac1N\tr(\B_N-z\I_N)^{-1}$. This new integral, whose value we name $\hat{P}_k$, is shown to be almost surely equal to $P_k$ in the large $N$ limit. It then suffices to evaluate $\hat{P}_k$, which is just a matter of residue calculus \cite{RUD86}.

\begin{figure} 
\centering
\includegraphics[]{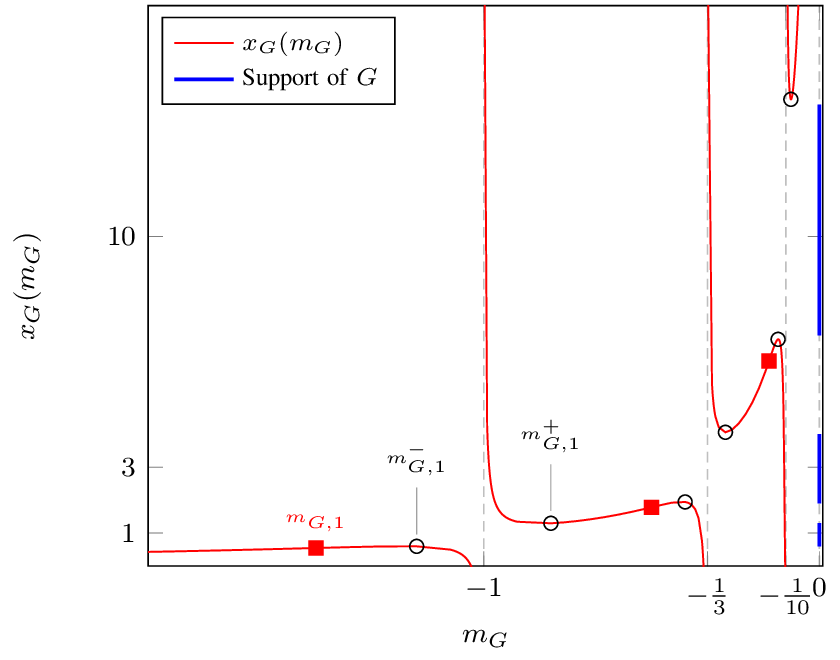}
  \caption{$x_G(m_G)$ for $m_G$ real, $\P$ diagonal composed of three evenly weighted masses in $1$, $3$ and $10$. Local extrema are marked in circles, inflexion points are marked in squares.}
  \label{fig:support}
\end{figure}

\begin{figure}
\centering
\includegraphics[]{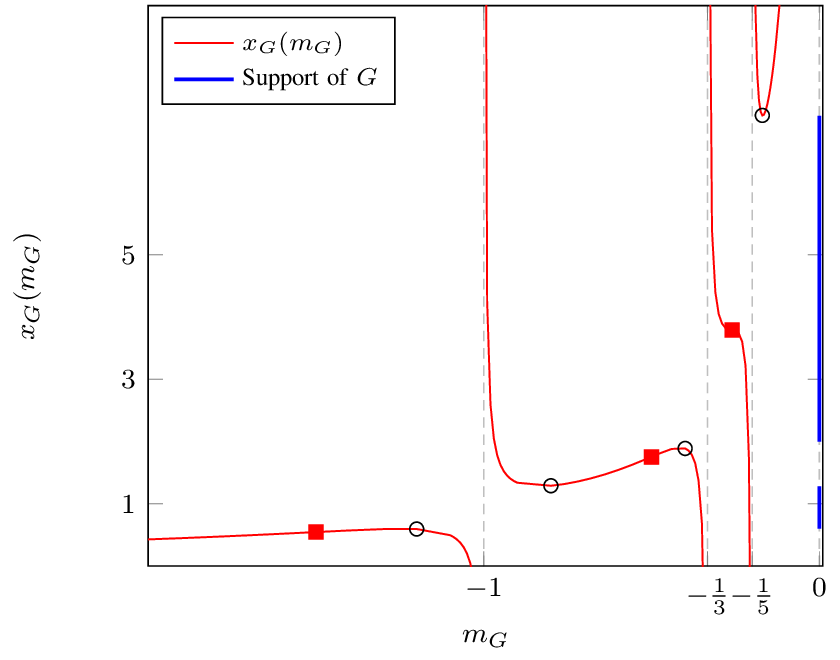}
  \caption{$x_G(m_G)$ for $m_G$ real, $\P$ diagonal composed of three evenly weighted masses in $1$, $3$ and $5$. Local extrema are marked in circles, inflexion points are marked in squares.}
  \label{fig:support2}
\end{figure}

\begin{figure} 
\centering
\includegraphics[]{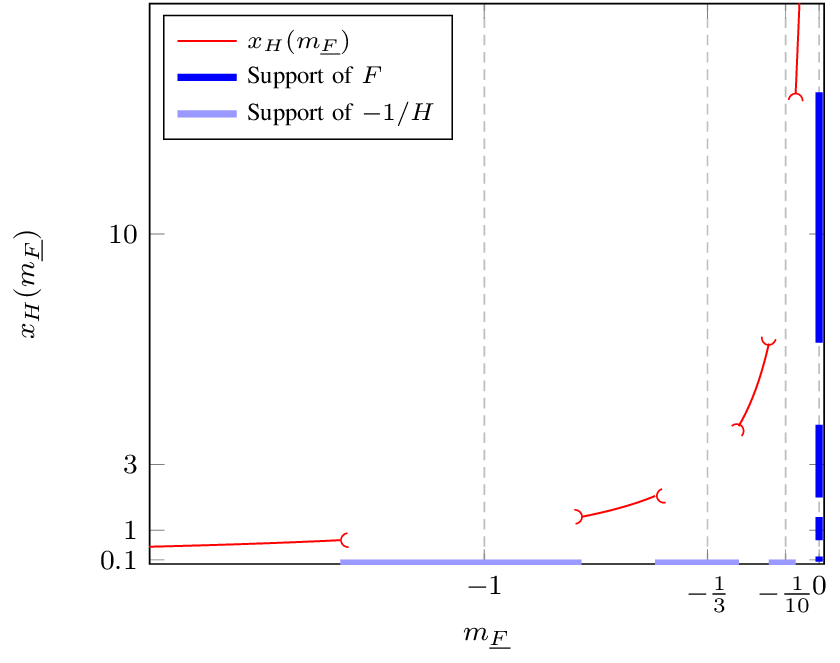}
  \caption{$x_{\uF}(\um)$ for $\um$ real, $\sigma^2=0.1$, $c=c_0=10$, $\P$ diagonal composed of three evenly weighted masses in $1$, $3$ and $10$. The support of $F$ is read on the vertical axis.}
  \label{fig:support_F}
\end{figure}

We start by determining the integration contour $\mathcal C_k$. For this, we first need to study the distributions $G$ and $F$ in more details.

\subsection{Properties of $G$ and $F$}
Let us introduce the following result on the l.s.d. of sample covariance matrices, borrowed from \cite{CHO95}
\begin{proposition}
  \label{prop:choi}
  Let $\A_n$ be defined as in Proposition \ref{prop:sil95}. Then the almost sure limiting Stieltjes transform $m_A(z)$ of the e.s.d. of $\A_n$, $z\in\CC^+$ admits a limit $m_A^\circ(x)$ when $z\to x\in \RR^\ast$. If $x$ is inside the support of $A$, then $m_A^\circ(x)$ is the only solution with positive imaginary part of the equation $x_A(m)=x$ in the variable $m$, with $x_A(m)$ defined, for $-1/m$ outside the support of $T$, as
  \begin{equation}
    \label{eq:xA}
    x_A(m) = -\frac1{m} + \frac1c\int \frac{t}{1+tm}dT(t),
  \end{equation}
while, if $x$ is outside the support of $A$, $m_A^\circ(x)$ is the only solution $m$ of $x_A(m)=x$ such that $x_A'(m)>0$. Moreover, if for some $m\in\RR$ such that $-1/m$ is outside the support of $T$, $x_A'(m)>0$, then $x_A(m)$ is outside the support of $A$.
\end{proposition}

The immediate corollary of Proposition \ref{prop:choi} is that the complementary of the support ${\rm Supp}(A)$ of $A$ is the set $\{x_A(m)\}$ for $-1/m$ outside the support of $T$ such that $x_A'(m)>0$, 
\begin{equation}
  {\rm Supp}(A) = \RR \setminus \left\{x ~\vert~ \exists m\in \RR,x=x_A(m),x_A'(m)>0 \right\}.
\end{equation}

\subsubsection{Support of $G$}
First consider the matrix $\H\P\H^\herm$, and let the function $x_G(m_G)$ be defined, for scalars $m_G\in\RR^\ast\setminus \{-1/P_1,\ldots,-1/P_K\}$, by
\begin{equation}
	\label{eq:xH_of_m1}
	x_G(m_G) = -\frac1{m_G} + \sum_{r=1}^K \frac1{c_r}\frac{P_r}{1+P_rm_G}.
\end{equation}

The function $x_G(m_G)$ is depicted in Figures \ref{fig:support} and \ref{fig:support2}, for the cases where $c_0=10$, $c_1=c_2=c_3$ and $(P_1,P_2,P_3)$ equal respectively $(1,3,10)$ and $(1,3,5)$. As expected by Proposition \ref{prop:choi}, $x_G(m_G)$ is increasing for $m_G$ such that $x_G(m_G)$ is outside the support of $G$. Note now that the function $x_G$ presents asymptotes in the positions $-1/P_1,\ldots,-1/P_K$,
\begin{align}
  \lim_{m_G\downarrow (-1/{P_i})} x_G(m_G) &= \infty \\
  \lim_{m_G\uparrow (-1/{P_i})} x_G(m_G)   &= -\infty,
\end{align}
and that $x_G(m_G)\to 0^+$ as $m_G\to -\infty$. Note also that, on its restriction to the set where it is non-decreasing, $x_G$ is increasing.\footnote{we say here that a function $f(x)$ is increasing if $x<x^\star \Rightarrow f(x)-f(x^\star)>0$; if $x<x^\star \Rightarrow f(x)-f(x^\star)\geq 0$, we say that $f(x)$ is non-decreasing.} To prove this, let $m_G$ and $m_G^\star$ be two distinct points such that $x_G(m_G)>0$ and $x_G(m_G^\star)>0$, and $m_G^\star<m_G<0$, we indeed have,\footnote{this proof is borrowed from the proof of \cite{MES08}, with different notations.}
\begin{align}
	\label{eq:m1m1star}
	x_G(m_G)-x_G(m^\star_G) &=\frac{m_G-m_G^\star}{m_Gm_G^\star} \nonumber \\
	&\times \left[1- \sum_{r=1}^K \frac1{c_r} \frac{P_r^2}{(P_r+\frac1{m_G})(P_r+\frac1{m_G^\star})} \right].
\end{align}

Noticing that, for $P_i>0$,
\begin{align}
	0 &< \left(\frac{P_i}{P_i+\frac1{m_G}}-\frac{P_i}{P_i+\frac1{m_G^\star}}\right)^2 \\
	&= \frac{P_i^2}{(P_i+\frac1{m_G})^2} + \frac{P_i^2}{(P_i+\frac1{m_G^\star})^2} - 2 \frac{P_i^2}{(P_i+\frac1{m_G})(P_i+\frac1{m_G^\star})},
\end{align}
we have
\begin{align}
  &\left(1- \sum_{r=1}^K \frac1{c_r} \frac{P_r^2}{(P_r+\frac1{m_G})^2}\right)+\left(1- \sum_{r=1}^K \frac1{c_r} \frac{P_r^2}{(P_r+\frac1{m_G^\star})^2}\right) \nonumber \\
  &< 2-2\sum_{r=1}^K \frac1{c_r} \frac{P_r^2}{(P_r+\frac1{m_G})(P_r+\frac1{m_G^\star})}.
\end{align}

Since we also have
\begin{align}
  x_G'(m_G) &= \frac1{m_G^2}\left[1 - \sum_{r=1}^K \frac1{c_r} \frac{P_r^2}{(P_r+\frac1{m_G})^2} \right] \geq 0 \\
  x_G'(m_G^\star) &= \frac1{(m_G^\star)^2}\left[1 - \sum_{r=1}^K \frac1{c_r} \frac{P_r^2}{(P_r+\frac1{m_G^\star})^2} \right] \geq 0,
\end{align}
we conclude that the term in brackets in \eqref{eq:m1m1star} is positive and then that $x_G(m_G)-x_G(m_G^\star)>0$. Hence $x_G$ is increasing on its restriction to the set where it is non-decreasing.

Notice also that $x_G$, both in Figures \ref{fig:support} and \ref{fig:support2}, has exactly one inflexion point on each open set $(-1/P_{i-1},-1/P_i)$, for $i\in\{1,\ldots,K\}$, with convention $P_0=0+$. This is proven by noticing that $x_G''(m_G)=0$ is equivalent to
  \begin{equation}
    \label{eq:xG2}
    \sum_{r=1}^K\frac1{c_r}\frac{P_r^3m_G^3}{(1+P_rm_G)^3}-1 = 0.
  \end{equation}

  Now, the left-hand side of \eqref{eq:xG2} has derivative along $m_G$,
  \begin{equation}
    3\sum_{r=1}^K \frac1{c_r}\frac{P_r^3m_G^2}{(1+P_rm_G)^4},
  \end{equation}
  which is always positive. Notice that the left-hand side of \eqref{eq:xG2} has asymptotes for $m_G=-1/P_i$ for all $i\in\{1,\ldots,K\}$, and has limits $0$ as $m_G\to 0$ and $1/c_0-1$ as $m_G\to -\infty$. If $c_0> 1$, Equation \eqref{eq:xG2} (and then $x_G''(m_G)=0$) therefore has a unique solution in $(-1/P_{i-1},-1/P_i)$ for all $i\in\{1,\ldots,K\}$. When $x_G$ is increasing somewhere on $(-1/P_{i-1},-1/P_i)$, the inflexion point, i.e., the solution to $x_G''(m_G)=0$, in $(-1/P_{i-1},-1/P_i)$ is necessarily found in the region where $x_G$ increases. If $c_0\leq 1$, the leftmost inflexion point may not exist.

  From the discussion above and Proposition \ref{prop:choi}, it is clear that the support of $G$ is divided into $\uK\leq K$ compact subsets $[x_{G,i}^-,x_{G,i}^+]$, $i\in\{1,\ldots,\uK\}$. Also, if $c_0>1$, $G$ has an additional mass in $0$ of probability $G(0)-G(0-)=(c_0-1)/c_0$; this mass will not be counted as a cluster in $G$. Observe that every $P_i$ can be uniquely mapped to a corresponding subset $[x_{G,j}^-,x_{G,j}^+]$ in the following fashion. The power $P_1$ is mapped onto the first cluster in $G$; we then have $1_G=1$. Then the power $P_2$ is either mapped onto the second cluster in $G$ if $x_G$ increases in the subset $(-1/P_1,-1/P_2)$, which is equivalent to saying that $x_G'(m_{G,2})>0$ for $m_{G,2}$ the only solution to $x_G''(m_G)=0$ in $(-1/P_1,-1/P_2)$; in this case, we have $2_G=2$ and the clusters associated to $P_1$ and $P_2$ in $G$ are distinct. Otherwise, if $x_G'(m_{G,2})\leq 0$, $P_2$ is mapped onto the first cluster in $F$; in this case, $2_G=1$. The latter scenario visually corresponds to the case when $P_1$ and $P_2$ engender ``overlapping clusters''. More generally, $P_j$, $j\in\{1,\ldots,K\}$, is uniquely mapped onto the cluster ${j_G}$ such that
\begin{equation}
  j_G = \# \left\{i\leq j ~\vert~ \min[x_G'(m_{G,i}),x_G'(m_{G,i+1})]>0 \right\},
\end{equation}
with convention $m_{G,K+1}=0$, which is exactly
\begin{equation}
  j_G = \# \left\{i\leq j ~\vert~ i \textmd{ satisfies Assumption \ref{ass:1}} \right\},
\end{equation}
when $c_0> 1$. If $c_0\leq 1$, $m_{G,1}$, the zero of $x_G''$ in $(-\infty,-1/P_1)$ may not exist. If $c_0<1$, we claim that $P_1$ cannot be evaluated (as was already observed in Remark \ref{rem:c0}). The special case when $c_0=1$ would require a restatement of Assumption \ref{ass:1} to handle the special case of $P_1$; this will however not be done, as it will turn out that Assumption \ref{ass:2} is violated for $P_1$ if $\sigma^2>0$, which we assume.

In the particular case of the power $P_k$ of interest in Theorem \ref{th:2}, because of Assumption \ref{ass:1}, $x_G'(m_{G,k})>0$. Therefore the index $\uk$ of the cluster associated to $P_k$ in $G$ satisfies $\uk=(k-1)_G+1$ (with convention $0_G=0$). Also, from Assumption \ref{ass:1}, $x_G'(m_{G,k+1})>0$. Therefore $(k+1)_G=\uk + 1$. In that case, we have that $P_k$ is the only power mapped to cluster $\uk$ in $G$, and then we have the required cluster separability condition.

\subsubsection{Support of $F$}
We now proceed to the study of $F$, the almost sure limit spectrum distribution of $\B_N$. In the same way as in the previous section, we have that the support of $\uF$ is fully determined by the function $x_\uF(\um)$, defined for $\um$ real, such that $-1/\um$ lies outside the support of $H$, by 
\begin{equation} 
  \label{eq:xF}
  x_\uF(\um) = -\frac1{\um} + \frac{1+c_0}{cc_0}\int \frac{t}{1+t\um}dH(t).
\end{equation}

Figure \ref{fig:support_F} depicts the function $x_\uF$ in the case of Figure \ref{fig:spectrum}, i.e., $K=3$, $P_1=1,P_2=3,P_3=10$, $c_1=c_2=c_3$, $c_0=10$, $c=10$, $\sigma^2=0.1$. Figure \ref{fig:support_F} has the peculiar behaviour that it does not have asymptotes as in Figure \ref{fig:support} where the population eigenvalue distribution was discrete. As a consequence, our previous derivations cannot be straightforwardly adapted to derive the spectrum separability condition. If $c_0>1$, note also, although it is not appearing in the abscissa range of Figure \ref{fig:support_F}, that there exist asymptotes in the position $\um=-1/\sigma^2$. This is due to the fact that $G(0)-G(0-)>0$, and therefore $H(\sigma^2)-H(\sigma^2-)>0$. We assume $c_0>1$ until further notice.

From Proposition \ref{prop:choi}, the support of $\uF$ is complementary to the set of real nonnegative $x$ such that $x=x_\uF(\um)$ and $x_\uF'(\um)>0$ for a certain real $\um$, with $x_\uF'(\um)$ given by
\begin{equation}
  x_\uF'(\um) = \frac1{\um^2} -\frac{1+c_0}{cc_0}\int \frac{t^2}{(1+t\um)^2}dH(t).
\end{equation}

Reminding that $H(t)=\frac{c_0}{c_0+1}G(t-\sigma^2)+\frac1{1+c_0}\delta(t)$, this can be rewritten
\begin{equation}
  \label{eq:xFp}
  x_\uF'(\um) = \frac1{\um^2} -\frac1c \int \frac{t^2}{(1+t\um)^2}dG(t-\sigma^2).
\end{equation}

It is still true that $x_\uF(\um)$, restricted to the set of $\um$ where $x_\uF'(\um)\geq 0$, is increasing. As a consequence, it is still true also that each cluster of $H$ can be mapped to a unique cluster in $\uF$. It is then possible to iteratively map the power $P_k$ onto cluster $\uk$ in $G$, as previously described, and to further map cluster $\uk$ in $G$ (which is also cluster $\uk$ in $H$) onto a unique cluster $k_F$ in $\uF$ (or equivalently in $F$). 

Therefore, a necessary and sufficient condition for the separability of the cluster associated to $P_k$ in $\uF$ reads
\begin{assumption}
  \label{ass:3}
  There exist two distinct real values $m_{\uF,\uk}^{(l)} < m_{\uF,\uk}^{(r)}$ such that 
  \begin{enumerate}
    \item $x_\uF'(m_{\uF,\uk}^{(l)})>0$, $x_\uF'(m_{\uF,\uk}^{(r)})>0$
    \item there exist $m_{G,k}^{(l)},m_{G,k}^{(r)}\in\RR$ such that $x_G(m_{G,k}^{(l)})=-1/m_{\uF,\uk}^{(l)}-\sigma^2$ and $x_G(m_{G,k}^{(r)})=-1/m_{\uF,\uk}^{(r)}-\sigma^2$ that satisfy 
      \begin{enumerate}
	\item $x_G'(m_{G,k}^{(l)})>0$, $x_G'(m_{G,k}^{(r)})>0$,
	\item \label{item:H} and
	  \begin{align}
	    P_{k-1} < -\frac1{m_{G,k}^{(l)}} < P_k < -\frac1{m_{G,k}^{(r)}} < P_{k+1}
	  \end{align}
	  with the convention $P_0=0+$, $P_{K+1}=\infty$.
      \end{enumerate}
  \end{enumerate}
\end{assumption}

Assumption \ref{ass:3} states (i) that cluster $\uk$ in $G$ is distinct from clusters $(k-1)_G$ and $(k+1)_G$ (Item \ref{item:H}); this is another way of stating Assumption \ref{ass:1}, and (ii) that the points $m_{\uF,\uk}^{(l)}\triangleq -1/(x_G(m_{G,\uk}^{(l)})+\sigma^2)$ and $m_{\uF,\uk}^{(r)}\triangleq -1/(x_G(m_{G,\uk}^{(r)})+\sigma^2)$ (which lie on either side of cluster $\uk$ in $H$) have respective images $x_\uuk^{(l)}\triangleq x_\uF(m_{\uF,\uk}^{(l)})$ and $x_\uuk^{(r)}\triangleq x_\uF(m_{\uF,\uk}^{(r)})$ by $x_\uF$, such that $x_\uF'(m_{\uF,\uk}^{(l)})>0$ and $x_\uF'(m_{\uF,\uk}^{(r)})>0$, i.e., $x_\uuk^{(l)}$ and $x_\uuk^{(r)}$ lie outside the support of $\uF$, on either side of cluster $\uuk$.

However, Assumption \ref{ass:3}, be it a necessary and sufficient condition for the separability of cluster $\uuk$, is difficult to exploit in practice. Indeed, it is not satisfactory to require the verification of the existence of such $m_{\uF,\uk}^{(l)}$ and $m_{\uF,\uk}^{(r)}$. More importantly, the computation of $x_\uF$ requires to know $H$, which is only fully accessible through the non-convenient inverse Stieltjes transform formula
\begin{equation}
  H(x) = \frac1\pi \lim_{y\to 0} \int_{-\infty}^x m_H(t+iy) dt.
\end{equation}

Instead of Assumption \ref{ass:3}, we derive here a sufficient condition for cluster separability in $\uF$. Notice from the clustering of $G$ into $\uK$ clusters plus a mass at zero that \eqref{eq:xFp} becomes
\begin{align}
  \label{eq:xFpp}
  x_\uF'(\um) &= \frac1{\um^2} -\frac1c \sum_{r=1}^\uK\int_{x_{G,r}^-}^{x_{G,r}^+} \frac{t^2}{(1+t\um)^2}dG(t-\sigma^2)  \nonumber \\ 
  & -\frac{c_0-1}{cc_0}\frac{\sigma^4}{(1+\sigma^2\um)^2},
\end{align}
where we remind that $[x_{G,i}^-,x_{G,i}^+]$ is the support of cluster $i$ in $G$, i.e., $x_{G,1}^-,x_{G,1}^+,\ldots,x_{G,\uK}^-,x_{G,\uK}^+$ are the images by $x_G$ of the $2\uK$ real solutions to $x_G'(m_G)=0$.

Observe now that the function $-t^2/(1+t\um)^2$, found in the integrals of \eqref{eq:xFpp}, has derivative along $t$
\begin{equation}
  \left(-\frac{t^2}{(1+t\um)^2} \right)^\prime = -\frac{2t}{(1+t\um)^4}(1+t\um)
\end{equation}
and is therefore strictly increasing when $\um<-1/t$ and strictly decreasing when $\um>-1/t$. For $\um\in(-1/(x_{G,i}^++\sigma^2),-1/x_{G,i+1}^-+\sigma^2)$, we then have the inequality
\begin{align}
  \label{eq:xxt}
  & x_\uF'(\um) \geq \frac1{\um^2} -\frac1c\left(\sum_{r=1}^i \frac{(x_{G,r}^++\sigma^2)^2}{(1+(x_{G,r}^++\sigma^2)\um)^2} \right. \nonumber \\
  &~~~ \left. + \sum_{r=i+1}^\uK \frac{(x_{G,r}^-+\sigma^2)^2}{(1+(x_{G,r}^-+\sigma^2)\um)^2} + \frac{c_0-1}{c_0}\frac{\sigma^4}{(1+\sigma^2\um)^2} \right).
\end{align}

Denote $f_i(\um)$ the right-hand side of \eqref{eq:xxt}. Through the inequality \eqref{eq:xxt}, we then fall back on a finite sum expression as in the previous study of the support of $G$. In that case, we can exhibit a sufficient condition to ensure the separability of cluster $\uuk$ from the neighboring clusters. Specifically, we only need to verify that $f_{\uk-1}(m_{\uF,\uk})>0$, with $m_{\uF,\uk}$ the single solution to $f_{\uk-1}'(\um)=0$ in the set $(-1/(x_{G,\uk-1}^++\sigma^2),-1/(x_{G,\uk}^-+\sigma^2))$, and $f_\uk(m_{\uF,\uk+1})>0$, with $m_{\uF,\uk+1}$ the unique solution to $f_{\uk}'(\um)=0$ in the set $(-1/(x_{G,\uk}^++\sigma^2),-1/(x_{G,\uk+1}^-+\sigma^2))$. This is exactly what Assumption \ref{ass:2} states.

Remember now that we assumed in this section $c_0>1$. If $c_0\leq 1$, then $0$ is in the support of $H$ and therefore the leftmost cluster in $F$, i.e., that attached to $\sigma^2$, is necessarily merged with that of $P_1$. This already discards the possibility of spectrum separation for $P_1$ and therefore $P_1$ cannot be estimated. It is therefore not necessary to update Assumption \ref{ass:1} for the particular case of $P_1$, when $c_0=1$.

Therefore, Assumptions \ref{ass:1} and \ref{ass:2} ensure that $(k-1)_F<k_F<(k+1)_F$, $k_F\neq 1$, and there exists a constructive way to derive the mapping $k\mapsto k_F$. We are now in position to determine the contour $\mathcal C_k$.


\subsection{Determination of $\mathcal C_k$}
From Assumption \ref{ass:2} and Proposition \ref{prop:choi}, there exist $x_\uuk^{(l)}$ and $x_\uuk^{(r)}$ outside the support of $F$, on either side of cluster $\uuk$, such that $\um(z)$ has limits $m_{\uF,\uk}^{(l)}\triangleq \um^\circ(x_\uuk^{(l)})$ and $m_{\uF,\uk}^{(r)}\triangleq \um^\circ(x_\uuk^{(r)})$, as $z\to x_\uuk^{(l)}$ and $z\to x_\uuk^{(r)}$, respectively, with $\um^\circ$ the analytic extension of $\um$ in the points $x_\uuk^{(l)}\in\RR$ and $x_\uuk^{(r)}\in\RR$. These limits $m_{\uF,\uk}^{(l)}$ and $m_{\uF,\uk}^{(r)}$ are on either side of cluster $\uk$ in the support of $-1/H$, and therefore $-1/m_{\uF,\uk}^{(l)}-\sigma^2$ and $-1/m_{\uF,\uk}^{(l)}-\sigma^2$ are on either side of cluster $\uk$ in the support of $G$. 

Consider any continuously differentiable complex path $\uuGk$ with endpoints $x_\uuk^{(l)}$ and $x_\uuk^{(r)}$, and interior points of positive imaginary part. We define the contour $\uuContour$ as the union of $\uuGk$ oriented from $x_\uuk^{(l)}$ to $x_\uuk^{(r)}$ and its complex conjugate $\uuGkast$ oriented backwards from $x_\uuk^{(r)}$ to $x_\uuk^{(l)}$. The contour $\uuContour$ is clearly continuous and piecewise continuously differentiable. Also, the support of cluster $\uuk$ in $\uF$ is completely inside $\uuContour$, while the supports of the neighboring clusters are away from $\uuContour$. The support of cluster $\uk$ in $H$ is then inside $-1/\um(\uuContour)$,\footnote{we slightly abuse notations here and should instead say that the support of cluster $\uk$ in $H$ is inside the contour described by the image by $-1/\um$ of the restriction to $\CC^+$ and $\CC^-$ of $\uuContour$, continuously extended to $\RR$ in the points $-1/m_{\uF,\uk}^{(l)}$ and $-1/m_{\uF,\uk}^{(r)}$.} and therefore the support of cluster $\uk$ in $G$ is inside $\uContour\triangleq -1/\um(\uuContour)-\sigma^2$. Since $\um$ is continuously differentiable on $\CC\setminus \RR$ (it is in fact holomorphic there \cite{CHO95}) and has limits in $x_\uuk^{(l)}$ and $x_\uuk^{(r)}$, $\uContour$ is also continuous and piecewise continuously differentiable. Going one last step in this process, we finally have that $P_k$ is inside the contour $\Contour \triangleq -1/m_G(\uContour)$, while $P_i$, for all $i\neq k$, is outside $\Contour$. Since $m_G$ is also holomorphic on $\CC\setminus \RR$ and has limits in $-1/\um^\circ(x_\uuk^{(l)})-\sigma^2$ and $-1/\um^\circ(x_\uuk^{(r)})-\sigma^2$, $\Contour$ is a continuous and piecewise continuously differentiable complex path, which is sufficient to perform complex integration \cite{RUD86}.

The contours $\mathcal C_1,\mathcal C_2,\mathcal C_3$ originating from circular integration contours ${\mathcal C}_{F,k}$ of diameter $[x_\uuk^{(l)},x_\uuk^{(r)}]$, $k\in\{1,2,3\}$, for the case of Figure \ref{fig:spectrum}, are depicted in Figure \ref{fig:contour}. The points $x_\uuk^{(l)}$ and $x_\uuk^{(r)}$ for $\uuk\in\{1,2,3\}$ are taken to be $x_\uuk^{(l)}=x_\uF(m_{\uF,\uk})$, $x_\uuk^{(r)}=x_\uF(m_{\uF,\uk+1})$, with $m_{\uF,i}$ the real root of $f_i'(\um)=0$ in $(-1/(x_{G,i-1}^++\sigma^2),-1/(x_{G,i}^-+\sigma^2))$ when $i\in\{1,2,3\}$, and we take the convention $m_{G,4}=-1/(15+\sigma^2)$.

Recall now that $P_k$ was defined as
\begin{equation}
	\label{eq:Pkintbis}
	P_k = c_k\frac1{2\pi i}\oint_{\Contour}\sum_{r=1}^{K}\frac1{c_r}\frac{\omega}{P_r-\omega}d\omega.
\end{equation}

With the variable change $\omega=-1/m_G(t)$, this becomes
\begin{align} 
  P_k = \frac{c_k}{2\pi i}\oint_{\uContour}&\left(m_G(t)\left[-\frac1{m_G(t)}+\sum_{r=1}^K\frac1{c_r}\frac{P_r}{1+P_rm_G(t)}\right] \right. \nonumber \\ 
  & \left. +\frac{c_0-1}{c_0} \right)\frac{m_G'(t)}{m_G(t)^2}dt.
\end{align}

From Equation \eqref{eq:m1}, this simplifies into
\begin{equation}
	\label{eq:Pkint_m1}
	P_k = \frac{c_k}{c_0}\frac1{2\pi i}\oint_{\uContour}\left(c_0 tm_G(t)+(c_0-1) \right)\frac{m_G'(t)}{m_G(t)^2}dt.
\end{equation}

Using \eqref{eq:z_of_m1_of_m} and proceeding to the further change of variable $t=-1/\um(z)-\sigma^2$, \eqref{eq:Pkint_m1} becomes
\begin{align}
	\label{eq:Pkint_m}
	P_k &= \frac{c_k}{2\pi i}\oint_{\uuContour} \left[ \left(\frac1{\um(z)}+\sigma^2\right)z\um(z)m_F(z)+\frac{c_0-1}{c_0} \right] \nonumber \\ &\times \frac{-\um(z)m_F(z)-z\um'(z)m_F(z)-z\um(z)m_F'(z)}{z^2\um(z)^2m_F(z)^2} dz \\
	&= \frac{c_k}{2\pi i}\oint_{\uuContour} \left[ \left(1+\sigma^2\um(z)\right)+\frac{c_0-1}{c_0}\frac1{z m_F(z)} \right] \nonumber \\ &\times \left[-\frac1{z\um(z)}-\frac{\um'(z)}{\um(z)^2}-\frac{m_F'(z)}{m_F(z)\um(z)}\right]dz.
\end{align}

This whole process of variable changes allowed us to describe $P_k$ as a function of $m_F(z)$, the Stieltjes transform of the almost sure limiting spectral distribution of $\B_N$, as $N\to \infty$. It then remains to exhibit a relation between $P_k$ and the empirical spectral distribution of $\B_N$ for finite $N$. This is to what the subsequent section is dedicated to.

\begin{figure}
  \centering
\includegraphics[]{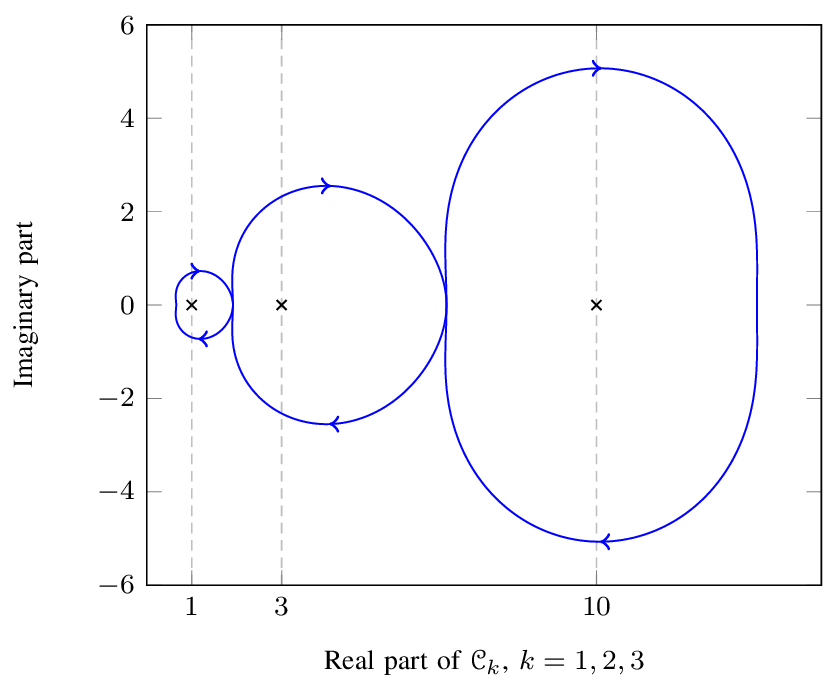}
  \caption{Integration contours $\mathcal C_{F,1}$, $\mathcal C_{F,2}$ and $C_{F,3}$, for $c=10$, $c_0=10$, $P_1=1$, $P_2=3$, $P_3=10$.}
  \label{fig:contour}
\end{figure}

\subsection{Evaluation of $\hat{P}_k$}
Let us now define $\hat{m}_F(z)$ and $\hum(z)$ as the Stieltjes transforms of the empirical eigenvalue distributions of $\B_N$ and $\uB_N$, respectively, i.e.,
\begin{equation}
	\label{eq:hatm}
	\hat{m}_F(z)=\frac1N\sum_{i=1}^N\frac1{\lambda_i-z}
\end{equation}
and 
\begin{equation}
\hum(z) = \frac{N}M\hat{m}_F(z) - \frac{M-N}M\frac1z.
\end{equation}

Instead of going further with \eqref{eq:Pkint_m}, define $\hat{P}_k$, the ``empirical counterpart'' of $P_k$, as
\begin{align}
	\label{eq:hPkint_m}
	\hat{P}_k = \frac{n}{n_k}&\frac1{2\pi i}\oint_{\uuContour} \left[ \frac{N}n\left(1+\sigma^2\hum(z)\right)+\frac{N-n}n\frac1{z \hat{m}_F(z)} \right] \nonumber \\ &\times \left[-\frac1{z\hum(z)}-\frac{\hum'(z)}{\hum(z)^2}-\frac{\hat{m}_F'(z)}{\hat{m}_F(z)\hum(z)}\right]dz.
\end{align}

The integrand can then be expanded into nine terms, for which residue calculus \cite{RUD86} can easily be performed. Denote first $\eta_1,\ldots,\eta_N$ the $N$ real roots of $\hat{m}_F(z)=0$ and $\mu_1,\ldots,\mu_N$ the $N$ real roots of $\hum(z)=0$. We identify three sets of possible poles for the nine aforementioned terms: (i) the set $\{\lambda_1,\ldots,\lambda_N\}\cap [x_\uuk^{(l)},x_\uuk^{(r)}]$, (ii) the set $\{\eta_1,\ldots,\eta_N\}\cap [x_\uuk^{(l)},x_\uuk^{(r)}]$ and (iii) the set $\{\mu_1,\ldots,\mu_N\}\cap [x_\uuk^{(l)},x_\uuk^{(r)}]$. For $M\neq N$, the full calculus leads to
	\begin{align}
	  \label{eq:residue_calc}
	  \hat{P}_k &= \frac{NM}{n_k(M-N)} \left[ \sum_{\substack{1\leq i\leq N \\ x_\uuk^{(l)} \leq \eta_i \leq x_\uuk^{(r)}}} \eta_i- \sum_{\substack{1\leq i\leq N \\ x_\uuk^{(l)} \leq \mu_i \leq x_\uuk^{(r)} }} \mu_i \right] \nonumber \\
	  & + \frac{N}{n_k}\left[ \sum_{\substack{1\leq i\leq N \\ x_\uuk^{(l)} \leq \eta_i \leq x_\uuk^{(r)}}} \sigma^2 - \sum_{\substack{1\leq i\leq N \\ x_\uuk^{(l)} \leq \lambda_i \leq x_\uuk^{(r)}}} \sigma^2 \right] \nonumber \\
	  & + \frac{N}{n_k}\left[ \sum_{\substack{1\leq i\leq N \\ x_\uuk^{(l)} \leq \mu_i\leq x_\uuk^{(r)}}} \sigma^2 - \sum_{\substack{1\leq i\leq N \\ x_\uuk^{(l)} \leq \lambda_i  \leq x_\uuk^{(r)}}} \sigma^2 \right] .
	\end{align}

	Details are given in Appendix \ref{app:residue}. Now, we know from Theorem \ref{th:1} that $\hat{m}_F(z)\asto m_F(z)$ and $\hum(z)\asto \um(z)$ as $N\to\infty$. Observing that the integrand in \eqref{eq:hPkint_m} is uniformly bounded on the compact $\uuContour$, the dominated convergence theorem \cite{BIL08} ensures $\hat{P}_k\asto P_k$.

	To go further, we now need to determine which of $\lambda_1,\ldots,\lambda_N$, $\eta_1,\ldots,\eta_N$ and $\mu_1,\ldots,\mu_N$ lie inside $\mathcal C_{F,k}$. This requires a result of eigenvalue {\it exact separation} that extends the earlier results of \cite{SIL98,BAI99}, as follows
\begin{theorem}
	\label{th:bai}
	Let $\B_n=(1/n)\T_n^\oh\X_n\X_n^\herm\T_n^\oh\in\CC^{p\times p}$, where we assume the following conditions
	\begin{enumerate}
	  \item $\X_n\in\CC^{p\times n}$ has entries $x_{ij}$, $1\leq i\leq p$, $1\leq j\leq n$, extracted from a doubly infinite array $\{x_{ij}\}$ of independent variables, with zero mean and unit variance. 
	  \item \label{item:thBai2} There exist $K$ and a random variable $X$ with finite fourth order moment such that, for any $x>0$, 
	    \begin{equation}
	      \frac1{n_1n_2}\sum_{i\leq n_1,j\leq n_2} P( |x_{ij}|>x)\le K P(|X|>x)
	    \end{equation}
for any $n_1,n_2$.
	  \item \label{item:thBai3} There is a positive function $\psi(x)\uparrow\infty$ as $x\to\infty$, and $M>0$, such that 
	    \begin{equation}
	      \max_{ij}\EE|x_{ij}^2|\psi(|x_{ij}|)\le M.
	    \end{equation}
	  \item $p=p(n)$ with $c_n = p/n \to c > 0$ as $n\to \infty$.
	  \item \label{item:thBai4} For each $n$, $\T_n\in \CC^{p\times p}$ is Hermitian nonnegative definite, independent of $\{x_{ij}\}$, satisfying $H_n \triangleq F^{\T_n} \Rightarrow H$, $H$ a nonrandom probability distribution function, almost surely. $\T_n^\oh$ is any Hermitian square root of $\T_n$.
	  \item \label{item:thBai5} The spectral norm $\Vert \T_n\Vert$ of $\T_n$ is uniformly bounded in $n$ almost surely.
	  \item Let $a,b>0$, nonrandom, be such that, with probability one, $[a,b]$ lies in an open interval outside the support of $F^{c_n,H_n}$ for all large $n$, with $F^{y,G}$ defined to be the almost sure l.s.d. of $(1/n)\X_n^\herm\T_n\X_n$ when $H=G$ and $c=y$.
	\end{enumerate}
	
	Denote $\lambda_1^\Y\geq \ldots\geq \lambda_p^\Y$ the ordered eigenvalues of the Hermitian matrix $\Y\in\CC^{p\times p}$. Then, we have that
	\begin{enumerate}
	  \item $P(\textmd{no eigenvalue of }\B_n\textmd{ lies in }[a,b]\textmd{ for all large }n)=1$.
	  \item If $c(1-H(0))>1$, then $x_0$, the smallest value in the support of $F^{c,H}$, is positive, and with probability one, $\lambda_n^{\B_n}\to x_0$ as $n\to\infty$.
	  \item If $c(1-H(0))\leq1$, or $c(1-H(0))>1$ but $[a,b]$ is not contained in $[0,x_0]$, then $m_{F^{c,H}}(a)<m_{F^{c,H}}(b)<0$. With probability one, there exists, for all $n$ large, an index $i_n\geq 0$ such that $\lambda_{i_n}^{\T_n}>-1/m_{F^{c,H}}(b)$ and $\lambda_{i_n+1}^{\T_n}>-1/m_{F^{c,H}}(a)$ and we have
	    \begin{equation}
	      \label{eq:separation}
	      P(\lambda_{i_n}^{\B_n}>b \textmd{ and } \lambda_{i_n+1}^{\B_n} < a \textmd{ for all large }n)=1.
	    \end{equation}
	  \end{enumerate}
\end{theorem}

Theorem \ref{th:bai} is proven in Appendix \ref{app:bai}.

To apply Theorem \ref{th:bai} to $\uB_N$ in our scenario, we need to ensure all assumptions are met. Only Items \ref{item:thBai2}-\ref{item:thBai5} need particular attention. In our scenario, the matrix $\X_n$ of Theorem \ref{th:bai} is $\left(\begin{smallmatrix} \X \\ \W \end{smallmatrix}\right)$, while $\T_n$ is $\T \triangleq \left(\begin{smallmatrix} \H\P\H^\herm + \sigma^2\I_N & 0 \\ 0 & 0 \end{smallmatrix}\right)$. The latter has been proven to have almost sure l.s.d. $H$, so that Item \ref{item:thBai4} is verified. Also, from the result of \cite{SIL98} upon which Theorem \ref{th:bai} is based, there exists a subset of probability one in the probability space that engenders the $\T$ over which, for $n$ large enough, $\T$ has no eigenvalues in any closed set strictly outside the support of $H$; this ensures Item \ref{item:thBai5}. Now, from construction, $\X$ and $\W$ have independent entries of zero mean, unit variance, fourth order moment and are composed of at most $K+1$ distinct distributions, irrespectively of $M$. Denote $X_1,\ldots,X_d$, $d\leq K+1$, $d$ random variables distributed as those distinct distributions. Letting $X=|X_1|+\ldots+|X_d|$, we have that
  \begin{align}
   \frac1{n_1n_2}\sum_{i\leq n_1,j\leq n_2} P( |z_{ij}|>x)
   &\leq P\left(\sum_{i=1}^d|X_i|>x\right) \\ 
   &= P(|X|>x),
  \end{align}
  where $z_{ij}$ is the $(i,j)^{th}$ entry of $\left(\begin{smallmatrix} \X \\ \W \end{smallmatrix}\right)$. Since all $X_i$ have finite order four moments, so does $X$ and Item \ref{item:thBai2} is verified. From the same argument, Item \ref{item:thBai3} follows with $\phi(x)=x^2$. Theorem \ref{th:bai} can then be applied to $\uB_N$.

    The corollary of Theorem \ref{th:bai} applied to $\uB_N$ is that, with probability one, for $N$ sufficiently large, there will be no eigenvalue of $\B_N$ (or $\uB_N$) outside the support of $F$, and the number of eigenvalues inside cluster $k_F$ is exactly $n_k$. Since $\uuContour$ encloses cluster $\uuk$  and is away from the other clusters, $\{\lambda_1,\ldots,\lambda_N\}\cap [x_\uuk^{(l)},x_\uuk^{(r)}]=\{\lambda_i,i\in \mathcal N_k\}$ almost surely, for all large $N$. Also, for any $i\in \{1,\ldots,N\}$, it is easy to see from \eqref{eq:hatm} that $\hat{m}_F(z)\to \infty$ when $z\uparrow \lambda_i$ and $\hat{m}_F(z) \to -\infty$ when $z\downarrow \lambda_i$. Therefore $\hat{m}_F(z)=0$ has at least one solution in each interval $(\lambda_{i-1},\lambda_i)$, with $\lambda_0=0$, hence $\mu_1<\lambda_1<\mu_2<\ldots<\mu_N<\lambda_N$. This implies that, if $k_0$ is the index such that $\uuContour$ contains exactly $\lambda_{k_0},\ldots,\lambda_{k_0+(n_k-1)}$, then $\uuContour$ also contains $\{\mu_{k_0+1},\ldots,\mu_{k_0+(n_k-1)}\}$. The same result holds for $\eta_{k_0+1},\ldots,\eta_{k_0+(n_k-1)}$. When the indexes exist, due to cluster separability, $\eta_{k_0-1}$ and $\mu_{k_0-1}$ belong, for $N$ large, to cluster $\uuk-1$. We are then left with determining whether $\mu_{k_0}$ and $\eta_{k_0}$ are asymptotically found inside $\uuContour$.  

For this, we use the same approach as in \cite{MES08}, by noticing that, since $0$ is not included in $\Contour$, one has
	\begin{equation}
	  \frac1{2\pi i}\oint_{\Contour}\frac1\omega d\omega = 0.
	\end{equation}

	Performing the same changes of variables as above, we have that
\begin{align}
  \label{eq:zero}
&\oint_{\uuContour} \frac{-\um(z)m_F(z)-z\um'(z)m_F(z)-z\um(z)m_F'(z)}{z^2\um(z)^2m_F(z)^2}dz \nonumber \\
&= 0.
\end{align}

For $N$ large, the dominated convergence theorem ensures again that the left-hand side of the \eqref{eq:zero} is close to
\begin{equation}
  \label{eq:hzero}
  \oint_{\uuContour} \frac{-\hum(z)\hm(z)-z\hum'(z)\hm(z)-z\hum(z)\hm'(z)}{z^2\hum(z)^2\hm(z)^2}dz.
\end{equation}

Residue calculus of \eqref{eq:hzero} then leads to
\begin{equation}
  \label{eq:h2zero}
  \left[\sum_{\substack{1\leq i\leq N \\ \lambda_i\in [x_\uuk^{(l)},x_\uuk^{(r)}]}} 2 - \sum_{\substack{1\leq i\leq N \\ \eta_i\in [x_\uuk^{(l)},x_\uuk^{(r)}]}} 1 - \sum_{\substack{1\leq i\leq N \\ \mu_i\in [x_\uuk^{(l)},x_\uuk^{(r)}]}} 1\right] \asto 0.
\end{equation}

Since the cardinalities of $\{i, \eta_i\in [x_\uuk^{(l)},x_\uuk^{(r)}]\}$ and $\{i,\mu_i\in [x_\uuk^{(l)},x_\uuk^{(r)}]\}$ are at most $n_k$, \eqref{eq:h2zero} is satisfied only if both cardinalities equal $n_k$ in the limit. As a consequence, $\mu_{k_0}\in [x_\uuk^{(l)},x_\uuk^{(r)}]$ and $\eta_{k_0}\in [x_\uuk^{(l)},x_\uuk^{(r)}]$. For $N$ large, $N\neq M$, this allows us to simplify \eqref{eq:residue_calc} into 
	\begin{align}
	  \label{eq:residue_calc2}
		\hat{P}_k &= \frac{NM}{n_k(M-N)} \sum_{\substack{1\leq i\leq N \\ \lambda_i\in \mathcal N_k}} (\eta_i-\mu_i) 
	\end{align}
	 with probability one. The same reasoning holds for $M=N$. This is our final relation.

	It now remains to show that the $\eta_i$ and the $\mu_i$ are the eigenvalues of $\diag(\blambda)-\frac1N\sqrt{\blambda}\sqrt{\blambda}^\trans$ and $\diag(\blambda)-\frac1M\sqrt{\blambda}\sqrt{\blambda}^\trans$ respectively. For this, we need the following lemma,
	\begin{lemma}
		\label{le:1}
		Let $\A\in\RR^{N\times N}$ be diagonal with entries $\lambda_1,\ldots,\lambda_N$, and let $\y\in\RR^N$. Then the eigenvalues of $\A-\y\y^\herm$ are the $N$ real solutions of the following equation in $x$,
		\begin{equation}
			\sum_{i=1}^N \frac{y_i^2}{\lambda_i-x} = 1.
		\end{equation}
	\end{lemma}
	\begin{proof}
Let $\lambda$ be an eigenvalue of $\A-\y\y^\herm$. For a certain non-zero vector $\x\in \CC^N$, we then have the equivalent relations
\begin{align}
(\A-\y\y^\herm)\x &=\lambda \x, \\
(\A-\lambda \I_N)\x &= \y^\herm \x \y,\\
\x &= \y^\herm \x(\A-\lambda \I_N)^{-1}\y, \\
\y^\herm \x &=\y^\herm \x \y^\herm(\A-\lambda \I_N)^{-1}\y, \\
1 &= \y^\herm(\A-\lambda \I_N)^{-1}\y.
\end{align}

Since $\A$ is diagonal, denoting $\e_i\in \CC^N$ the vector such that $e_{i,j}=\delta_i^j$, we finally have
\begin{equation}
	\sum_{i=1}^N \frac{(\y^\herm \e_i)^2}{\lambda_i - \lambda} = 1.
\end{equation}
	\end{proof}

	Applying Lemma \ref{le:1} to $\A=\diag{\blambda}$ and $\y=\sqrt{\frac1N\blambda}$, we find that the eigenvalues of $\diag(\blambda)-\frac1N\sqrt{\blambda}\sqrt{\blambda}^\trans$ are the solutions of  
		\begin{equation}
			\sum_{i=1}^N \frac{\frac1N\lambda_i}{\lambda_i-x} = 1,
		\end{equation}
which is equivalent to 
		\begin{equation}
			\frac1N\sum_{i=1}^N \frac1{\lambda_i-x} = 0,
		\end{equation}
		whose solutions are by definition $\eta_1,\ldots,\eta_N$. The same argument applies similarly to $\mu_1,\ldots,\mu_N$. Incidentally, this remark was already noticed in \cite{GRE09}.

We end this section by a short discussion on the consequences of Theorem \ref{th:2}.

\subsection{Discussion}
\label{sec:discussion}
Theorem \ref{th:2} states that, under spectrum separability condition for all $P_k$, $k\in\{1,\ldots,K\}$, when $n_1,\ldots,n_K$ are known {\it a priori} to the receiver, then $\hat{P}_1,\ldots,\hat{P}_K$ are consistent estimators for $P_1,\ldots,P_K$. Now, in practice, it is rare that $n_1,\ldots,n_K$ and even $K$ are {\it a priori} known to the receiver. However, if separability is assumed, then one can estimate simultaneously $K,n_1,\ldots,n_K$ and $P_1,\ldots,P_K$. This is performed by (i) determining the clusters of the empirical eigenvalues of $\B_N$, which determines $K$, (ii) counting the number of eigenvalues in each cluster to determine the multiplicities $n_1,\ldots, n_K$ and (iii) evaluating $\hat{P}_1,\ldots,\hat{P}_K$ from Theorem \ref{th:2}.

However, step (i) may not be obvious. In particular, when the total number $n$ of transmit antennas is small, when the typical cluster size is large or when the inter-cluster spacing is small, it is non-trivial to determine what eigenvalues form a cluster. To solve this critical issue, studies are being currently carried out that aim to determine second order statistics of $F^{\B_N}$. Thanks to second order statistics on $F^{\B_N}$, it will be possible to design estimators of $P_1,\ldots,P_K$ that take into account the probability of $\B_N$ being an appropriate model for the estimated $\hat{P}_1,\ldots,\hat{P}_{\hat{K}}$ for every hypothesis $\hat{K}$ for the number of transmit source and every hypothesis $(\hat{n}_1,\ldots,\hat{n}_{\hat{K}})$ for the number of antennas for each of these sources. We hereafter provide an alternative {\it ad-hoc} technique to partially solve the problem of determining $K$ and $n_1,\ldots,n_K$ based on Theorems \ref{th:1} and \ref{th:2}.

In the following, we assume for readability that we know the number $K$ of transmit sources (taken large enough to cover all possible hypotheses), some having possibly $0$ transmit antennas. The approach consists in the following steps:
\begin{enumerate}
  \item we first identify a set of plausible hypotheses for $n_1,\ldots,n_K$. This can be performed by inferring clusters based on the spacing between consecutive eigenvalues: if the distance between neighboring eigenvalues is more than a threshold, then we add an entry for a possible cluster separation in the list of all possible positions of cluster separation. From this list, we create all possible $K$-dimensional vectors of eigenvalue clusters. Obviously, the choice of the threshold is critical to reduce the number of hypotheses to be tested;
  \item for each $K$-dimensional vector with number of antennas $\hat{n}_1,\ldots,\hat{n}_K$, we use Theorem \ref{th:2} in order to obtain estimates of the $\hat{P}_1,\ldots,\hat{P}_K$ (some being possibly null);
  \item based on these estimates, we compare the e.s.d. $F^{\B_N}$ of $\B_N$ to the distribution function $\hat{F}$ defined as the l.s.d. of the matrix model $\hat{\Y}=\H\hat{\P}\X+\W$ with $\hat{\P}$ the diagonal matrix composed of $\hat{n}_1$ entries equal to $\hat{P}_1$, $\hat{n}_2$ entries equal to $\hat{P}_2$ etc. up to $\hat{n}_K$ entries equal to $\hat{P}_K$. The comparison can be performed based on different metrics. In the simulations carried hereafter, we consider as a metric the mean absolute difference between the Stieltjes transform of $F^{\B_N}$ and of $\hat{F}$ on the segment $[-1,-0.1]$.
\end{enumerate}

Note that the above process can bring an interesting feature linked to the cluster separability problem discussed along this article. Indeed, if two subsequent powers $P_i$ and $P_{i+1}$ are close to one another, then the separability condition of Assumptions \ref{ass:1} and \ref{ass:2} is not verified. If one knows $n_i$ and $n_{i+1}$ and blindly uses the estimator of Theorem \ref{th:2}, the result can be catastrophic as the estimator is unreliable. On the contrary, if $n_i$ and $n_{i+1}$ are unknown and one uses the above process, it is very likely that the distinct sources with close power will be assumed to be a single source with power equal to the estimate of $(P_i+P_{i+1})/2$ and embedded with $n_i+n_{i+1}$ antennas. For practical blind detection purposes in cognitive radios, this leads the secondary network to infer a number of transmit entities that is less than the effective number of transmitters. In general, this would not have serious consequences on the decisions made by the secondary network but this might at least reduce the capabilities of the secondary network to optimally overlay the licensed spectrum. Further work is also being carried out to go past the cluster separability assumption; specifically, methods for estimating the number of $P_i$ associated to any cluster $\uuj$ are under study.

\section{Simulations}
\label{sec:simus}

In this section, we provide simulation results to assess the performance of Theorem \ref{th:2} when $K$, and $n_1,\ldots,n_K$ are known, to compare this performance against alternative estimation methods and finally to evaluate the performance of the {\it ad-hoc} approach discussed in Section \ref{sec:discussion}. In order to underline some precise features of the advantages of our novel method, we will use two simulation models. The first model, already presented in Figure \ref{fig:spectrum}, involves a scenario with clear separation between clusters, while the second model will consider the case of co-located clusters.

The estimator of Theorem \ref{th:2} will be compared against two methods, which we describe below. 

\subsection{Alternative methods}

\subsubsection{Strongly consistent estimator for $M\gg N$ and $N\gg n$}
The first method is the classical estimator that assumes that the sample dimension $M$ is much larger than the sensor dimension $N$, while $N$ is much larger than the source dimension $n$. In this case, it is easy to see that the e.s.d. of $\B_N$ tends to a mass in $\sigma^2$. However, the first $n$ eigenvalues of $\B_N$ are asymptotically greater than $\sigma^2$ and it is also clear that the e.s.d. of the projection of $\B_N$ on the eigenspace associated to its largest $n$ eigenvalues tends to $K$ masses in $P_1+\sigma^2,\ldots,P_K+\sigma^2$. This leads to the strongly consistent estimator $\hat{P}_k^\infty$ of $P_k$ given by
\begin{equation}
  \label{eq:classicalest}
  \hat{P}_k^\infty = \frac1{n_k}\sum_{i\in\mathcal N_k}(\lambda_i-\hat{\sigma}^2),
\end{equation}
with
\begin{equation*}
  \hat{\sigma}^2 = \frac1{N-n}\sum_{i=1}^{N-n} \lambda_i
\end{equation*}
and we recall that $\lambda_1\leq \ldots\leq \lambda_N$ are the eigenvalues of $\B_N$. The strong consistence is with respect to the rates $n\to \infty$, $N/n\to \infty$ and $M/N\to\infty$. Note that we take an estimator for $\sigma^2$ instead of $\sigma^2$ itself in order to be coherent with Theorem \ref{th:2} which does not require any {\it a priori} information on $\sigma^2$. We will refer to this estimator as the {\it classical} method.

\subsubsection{Estimator based on strongly consistent moment estimates}
The second method is a technique issued from free probability theory, which is based on moments of the l.s.d. of $\B_N$. As such, we will refer to this method as the {\it moment} method. It consists in computing the first moments of the e.s.d. of $\B_N$, i.e., $\frac1N\tr \left(\frac1M\Y\Y^\herm\right)^k$, for $k=1,\ldots,K$, from which the {\it deconvolved} moments $\frac1n(n_1P_1^k+\ldots+n_KP_K^k)$ of $F^{\P}$ can be evaluated, see e.g., \cite{RYA07}. These estimated moments can be expressed as polynomials of the moments of $F^{\B_N}$, which is convenient from a practical point of view although it leads to serious shortcomings in terms of estimator accuracy. Indeed, small deviations in the low order moments of $F^{\B_N}$ around the corresponding moments of $F$ lead to large deviations in the estimation of the high order moments of $F^{\P}$.

One can then retrieve the vector $(\hat{P}_1^{\rm (mom)},\ldots,\hat{P}_K^{\rm (mom)})$ whose distribution function has for first $K$ moments the first $K$ estimated moments of $F^{\P}$. This is performed using Newton-Girard polynomial formulas \cite{SER00}, which boils down to finding the roots of a polynomial of order $K$. The value $\hat{P}_k^{\rm (mom)}$ is the estimate of $P_k$. Computing $\hat{P}_k^{\rm (mom)}$ requires in particular that $K,n_1,\ldots,n_K$ and $\sigma^2$ are known. The main shortcoming of the Newton-Girard inversion is that the polynomial to be solved may have purely imaginary roots. This issue, added to the deviations in the estimated moments, contribute to rather poor estimation accuracies unless the system dimensions are very large. However, as opposed to the classical method and the novel Stieltjes transform approach, the moment method does not require any assumption of cluster separability to be valid.
 

\subsection{Results}

\subsubsection{Cluster separability limit}

We start with a demonstration of the performance of the novel estimator with respect to the satisfaction of the cluster separability assumption. We consider the model presented in Figure \ref{fig:spectrum}, i.e., $K=3$, $P_1=1$, $P_2=3$, $P_3=10$, $n_1/n=n_2/n=n_3/n=1/3$ and $n/N=N/M=1/10$. The SNR, defined as ${\rm SNR}=1/\sigma^2$, ranges from $-15$ dB to $20$ dB. The entries of $\X$ are QPSK-modulated and those of $\H$ and $\W$ are Gaussian distributed. In Figure \ref{fig:MSE_individual}, we present simulation results in terms of normalized mean square error (NMSE) in the estimates of the individual $P_k$, both for $n=60$ and $n=6$. For future need, we define this system model with $n=6$ as Scenario (a). The NMSE for power $P_k$ is given by
\begin{equation}
	{\rm NMSE} = \EE \left[ \frac{(P_k-\hat{P}_k)^2}{P_k^2} \right],
\end{equation}
where the expectation is taken over the random realizations of the matrices $\H$, $\X$ and $\W$. 

Note how steep the mean square error curves increase below a given SNR value. This intuitively corresponds to the tipping point where the cluster separability assumptions are no longer verified. Especially here, this corresponds to the point where Assumption \ref{ass:2} no longer holds. Now, remembering the results of Figure \ref{fig:ass2}, observe that the horizontal line $c=10$ crosses the respective curves of validity of Assumption \ref{ass:2} around the SNR values where Figure \ref{fig:MSE_individual} shows steep curve increase. This indicates that our novel estimator is indeed inappropriate when Assumption \ref{ass:3} is not satisfied. This also validates the accuracy of Assumption \ref{ass:2}, which we recall is only a sufficient condition for cluster separability. Note also that, as long as cluster separation is achieved, the performance of the Stieltjes transform algorithm goes quickly down to a constant level (with respect to the SNR) which is a function of the amplitude of the values of $n$, $N$ and $M$.

\begin{figure} 
\centering
\includegraphics[]{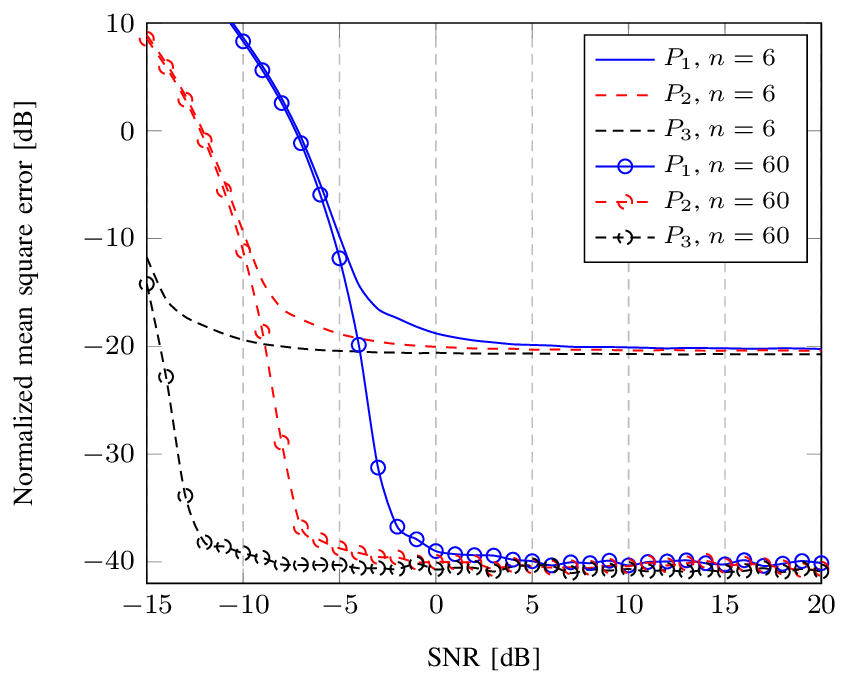}
  \caption{Normalized mean square error of individual powers $\hat{P}_1$, $\hat{P}_2$, $\hat{P}_3$, $P_1=1,P_2=3,P_3=10$, $n_1/n=n_2/n=n_3/n=1/3$ ,$n/N=N/M=1/10$, for $10,000$ simulation runs.}
  \label{fig:MSE_individual}
\end{figure}

\subsubsection{Performance comparison}

We first compare the classical method against the novel Stieltjes transform approach for Scenario (a). Under the hypotheses of this scenario, the ratios $c$ and $c_0$ equal $10$, leading therefore the classical detector to be almost asymptotically unbiased. We therefore suspect that the NMSE performance for both detectors is alike. This is described in Figure \ref{fig:eigen-inferenceMSE_individual}, which suggests as predicted that in the high SNR regime (when cluster separability is reached) the classical estimator performs similar to the Stieltjes transform method. However, it appears that a $3$ dB gain is achieved by the Stieltjes transform method around the position where cluster separability is no longer satisfied. This translates the fact that, when subsequent clusters tend to merge as $\sigma^2$ increases, the Stieltjes transform method manages to track the position of the powers $P_k$ while the classical method keeps assuming each $P_k$ is located at the center of cluster $k_F$. This observation is very similar to that made in \cite{MES08c}, where an improved MUSIC estimator is introduced that pushes further the SNR position where the performance of the classical MUSIC estimator decays significantly.

\begin{figure} 
\centering
\includegraphics[]{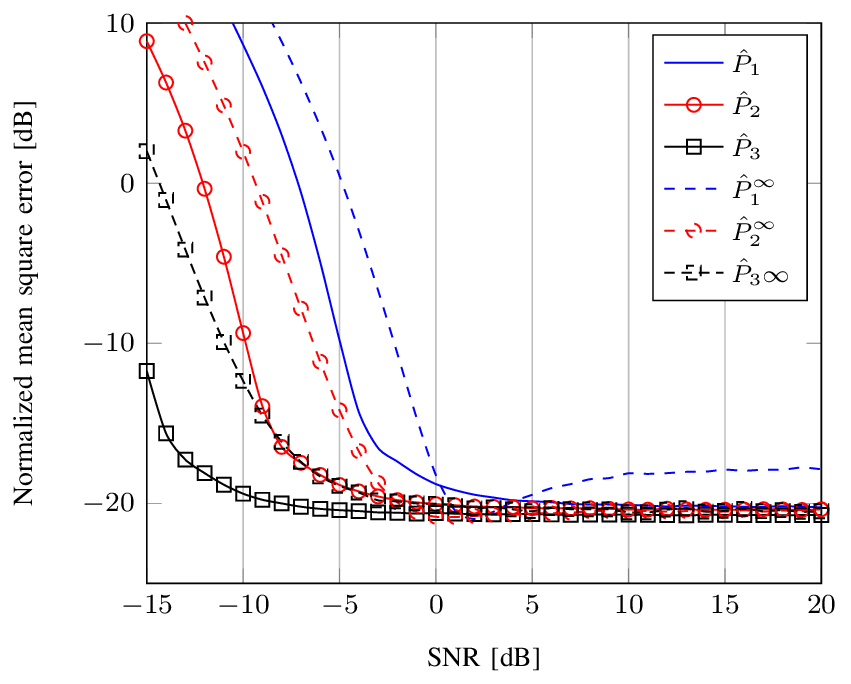}
\caption{Normalized mean square error of individual powers $\hat{P}_1$, $\hat{P}_2$, $\hat{P}_3$, $P_1=1,P_2=3,P_3=10$, $n_1/n=n_2/n=n_3/n=1/3$ ,$n/N=N/M=1/10$, $n=6$. Comparison between classical and Stieltjes transform approach.}
\label{fig:eigen-inferenceMSE_individual}
\end{figure}

We now consider another model, for which the classical estimator is largely biased. We now take $K=3$, $P_1=1/16$, $P_2=1/4$, $P_3=1$, $n_1/n=n_2/n=n_3/n=1/3$ and $n=12$, $N=24$ and $M=128$. The entries of $\X$ are still QPSK-modulated while the entries of $\H$ and $\W$ are still independent standard Gaussian. This model is further referred to as Scenario (b). We first compare the performance of the classical, Stieltjes transform and moment estimators for an SNR of $20$ dB. Figure \ref{fig:applications_estimation_power_comparisondf} depicts the distribution function of the estimated powers in logarithmic scale. The Stieltjes transform method appears here to be very precise and seemingly unbiased. On the opposite, the classical method, with a slightly smaller variance shows a large bias as was anticipated. As for the moment method, it shows rather accurate performance for the stronger estimated power, but proves very inaccurate for smaller powers. This entails from the inherent shortcomings of the moment method. The performance of the estimator $\hat{P}_k'$ will be commented in Section \ref{sec:sim_jointest}.

We then focus on the estimate for the larger power $P_3$ and take now the SNR to range from $-15$ to $30$ dB under the same conditions as previously and for the same estimators. The NMSE for the estimators of $P_3$ is depicted in Figure \ref{fig:eigen-inferenceMSE_individual2}. The curve marked with squares will be commented in Section \ref{sec:sim_jointest}. As already observed in Figure \ref{fig:applications_estimation_power_comparisondf}, in the high SNR regime, the Stieltjes transform estimator outperforms both alternative methods. We also notice the SNR gain achieved by the Stieltjes transform approach with respect to the classical method in the low SNR regime, as already observed in Figure \ref{fig:eigen-inferenceMSE_individual}. However, it now turns out that in this low SNR regime, the moment method is gaining ground and outperforms both cluster-based methods. This is due to the cluster separability condition which is not a requirement for the moment approach. This indicates that much can be gained by the Stieltjes transform method in the low SNR regime if a more precise treatment of overlapping clusters is taken into account.

\begin{figure}
\centering
\includegraphics[]{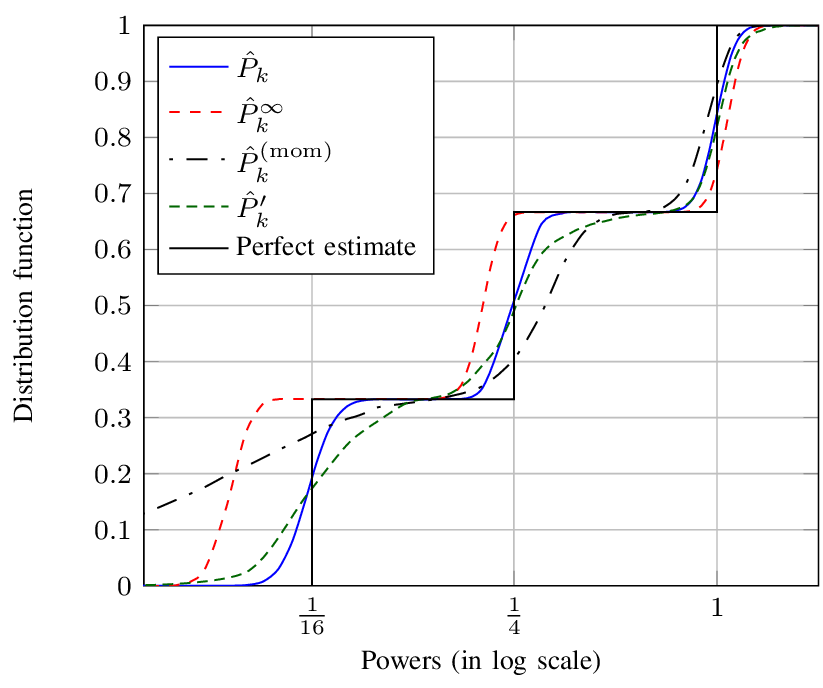}
\caption{Distribution function of the estimators $\hat{P}_k^\infty$, $\hat{P}_k$, $\hat{P}_k'$ and $\hat{P}^{\rm (mom)}_k$ for $k\in\{1,2,3\}$, $P_1=1/16$, $P_2=1/4$, $P_3=1$, $n_1=n_2=n_3=4$ antennas per user, $N=24$ sensors, $M=128$ samples and ${\rm SNR}=20$ dB. Optimum estimator shown in dashed lines.}
\label{fig:applications_estimation_power_comparisondf}
\end{figure}

\begin{figure} 
\centering
\includegraphics[]{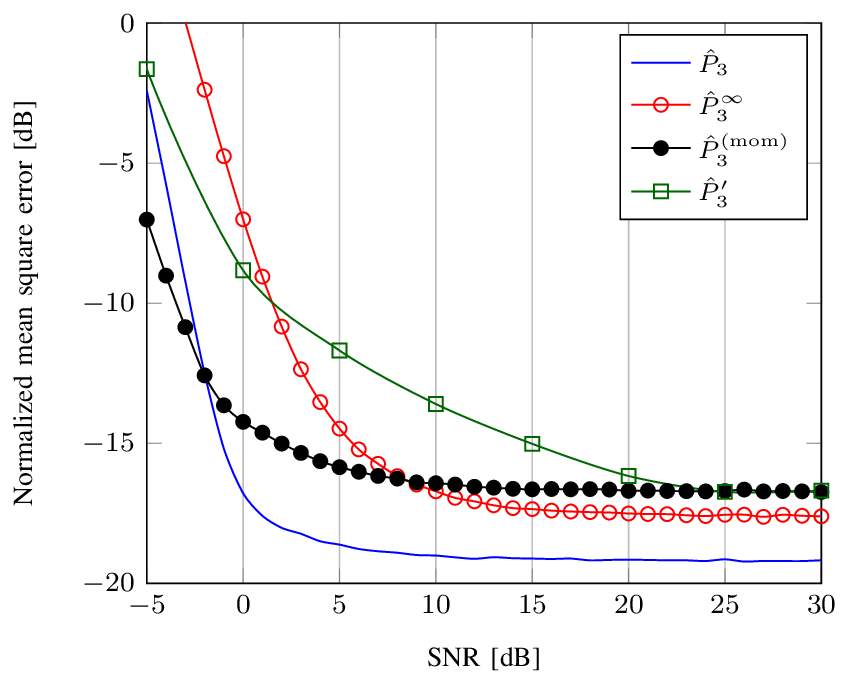}
\caption{Normalized mean square error of largest estimated power $P_3$, $P_1=1/16,P_2=1/4,P_3=1$, $n_1=n_2=n_3=4$ ,$N=24$, $M=128$. Comparison between classical, moment and Stieltjes transform approaches.}
\label{fig:eigen-inferenceMSE_individual2}
\end{figure}

\subsubsection{Joint estimation of $K$, $n_k$, $P_k$}
\label{sec:sim_jointest}
So far, we have assumed that the number of users $K$ and the number of antennas per user $n_k$ were perfectly known. As discussed in Section \ref{sec:discussion}, this may not be a strong assumption if it is known by advance how many antennas are systematically used by every source or if another mechanism, such as in \cite{HER07}, can provide this information. Nonetheless, these are in general strong assumptions to take. Based on the {\it ad-hoc} method described in Section \ref{sec:discussion}, we therefore provide the performance of our novel Stieltjes transform method in the high SNR regime when only $n$ is known; this assumption is less stringent as in the medium to high SNR regime, one can easily decide which eigenvalues of $\B_N$ belong to the cluster associated to $\sigma^2$ and which eigenvalues do not. We denote $\hat{P}_k'$ the estimator of $P_k$ when $K$ and $n_1,\ldots,n_K$ are unknown. We assume for this estimator that all possible combinations of $1$ to $3$ clusters can be generated from the $n=6$ observed eigenvalues in Scenario (a) and that all possible combinations of $1$ to $3$ clusters with even cluster size can be generated from the $n=12$ eigenvalues of $\B_N$ in Scenario (b). For Scenario (a), the NMSE performance of the estimators $\hat{P}_k$ and $\hat{P}_k'$ is proposed in Figure \ref{fig:MSE_individual_unknown_K} for the SNR ranging from $5$ dB to $30$ dB. For Scenario (b), the distribution function of the inferred $\hat{P}_k'$ is depicted in Figure \ref{fig:applications_estimation_power_comparisondf}, while the NMSE performance for the inference of $P_3$ is proposed in Figure \ref{fig:eigen-inferenceMSE_individual2}; these are both compared against the classical, moment and Stieltjes transform estimator. We also indicate in Table \ref{tab:unknown_K} the percentage of correct estimation of the triplet $(n_1,n_2,n_3)$ for both Scenario (a) and (b). In Scenario (a), this amounts to $12$ such triplets that satisfy $n_k\geq 0$, $n_1+n_2+n_3=6$, while in Scenario (b), this corresponds to $16$ triplets that satisfy $n_k\in 2\NN$, $n_1+n_2+n_3=12$. Observe that the noise variance, assumed to be known {\it a priori} in this case, plays an important role with respect to the statistical inference of the $n_k$. 
In Scenario (a), for an SNR greater than $15$ dB, the correct hypothesis for the $n_k$ is almost always taken and the performance of the estimator is similar to that of the optimal estimator. In Scenario (b), the detection of the exact cluster separation is less accurate and the performance for the inference of $P_3$ saturates at high SNR to $-16$ dB of NMSE, against $-19$ dB when the exact cluster separation is known. It therefore seems that in the high SNR regime the performance of the Stieltjes transform detector is loosely affected by the absence of knowledge about the cluster separation. This statement is also confirmed by the distribution function of $\hat{P}_k'$ in Figure \ref{fig:applications_estimation_power_comparisondf}, which still outperforms the classical and moment methods. We underline again here that this is merely the result of an {\it ad-hoc} approach; this performance could be greatly improved if e.g., more is known about the second order statistics of $F^{\B_N}$.

\begin{table}   
  \label{tab:unknown_K}
  \begin{center}
  \begin{tabular}{c|cccccc}
    SNR & RCI (a)  & RCI (b) \\  
    \hline
$5$  dB & $0.8473$ & $0.1339$ \\
$10$ dB & $0.9026$ & $0.4798$ \\
$15$ dB & $0.9872$ & $0.4819$ \\
$20$ dB & $0.9910$ & $0.5122$ \\
$25$ dB & $0.9892$ & $0.5455$ \\
$30$ dB & $0.9923$ & $0.5490$
  \end{tabular}
\end{center}
\caption{Rate of correct inference (RCI) of the triplet $(n_1,n_2,n_3)$ for scenarios (a) and (b).}
\end{table}



\begin{figure} 
\centering
\includegraphics[]{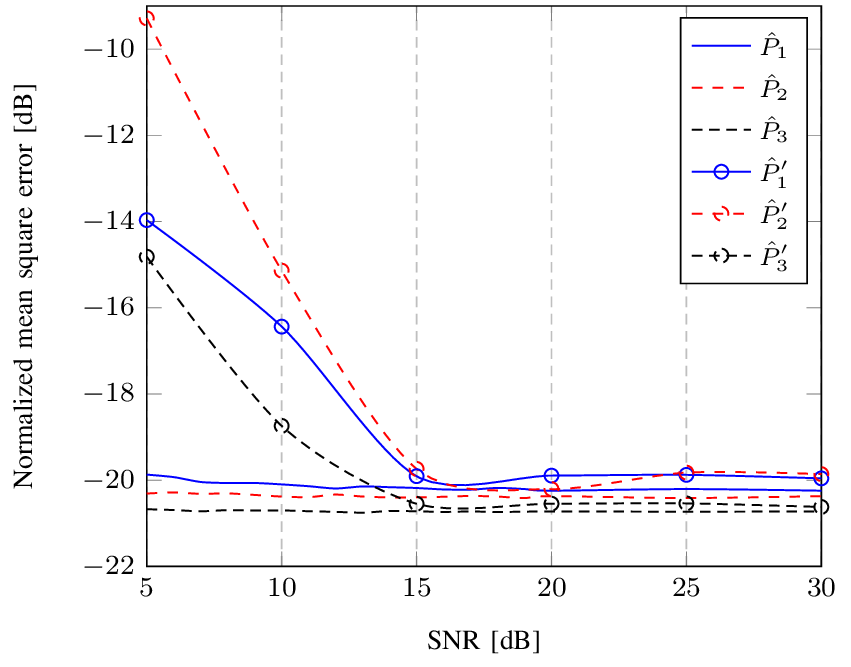}
  \caption{Normalized mean square error of individual powers $\hat{P}_1$, $\hat{P}_2$, $\hat{P}_3$ and $\hat{P}_1'$, $\hat{P}_2'$, $\hat{P}_3'$, $P_1=1,P_2=3,P_3=10$, $n_1/n=n_2/n=n_3/n=1/3$ ,$n/N=N/M=1/10$, $n=6$, $10,000$ simulation runs.}
  \label{fig:MSE_individual_unknown_K}
\end{figure}

\section{Conclusion}
\label{sec:conclusion}
In this paper, a blind multi-source power estimator was derived. Under the assumptions that the ratio between the number of sensors and the number of signals sources is not too small and the source transmit powers are sufficiently distinct from one another, we derived a method to infer the individual source powers if the number of sources is known, which was shown to outperform alternative estimation techniques in the medium to high SNR regime. We then briefly discussed the joint estimation of the number of transmit sources, the number of antennas of each source and the transmit powers, which appeared in simulation to perform well in the high SNR regime. The novel method is moreover computationally efficient and is particularly robust to small system dimensions. As such, it is particularly suited to the blind detection of primary mobile user in future cognitive radio networks.

\appendices

\section{Residue calculus}
\label{app:residue}
The integrand of $\hat{P}_k$ in \eqref{eq:hPkint_m} can be expanded as
\begin{align}
  &-\frac{N}{n}\frac1{z\hum(z)}-\frac{N}{n}\frac{\hum'(z)}{\hum(z)^2}-\frac{N}{n} \frac{\hm'(z)}{\hum(z)\hm(z)} \nonumber \\ 
  &-\frac{N-n}{n}\frac1{z^2\hum(z)\hm(z)}-\frac{N-n}{n}\frac{\hum'(z)}{z\hum(z)^2\hm(z)}\nonumber \\ 
  &-\frac{N-n}{n}\frac{\hm'(z)}{z\hm(z)^2\hum(z)} -\frac{N\sigma^2}{n}\frac1z-\frac{N\sigma^2}{n}\frac{\hum'(z)}{\hum(z)} \nonumber \\
  &-\frac{N\sigma^2}{n}\frac{\hm'(z)}{\hm(z)}
\end{align}

First assume the case $M\neq N$. Numbering the nine terms in order, we have that (1) has poles in $z\in\{\eta_1,\ldots,\eta_N\}$, where $\hum(z)=0$. Applying l'Hospital rule, all poles have order $1$ and the corresponding residues are 
\begin{equation}
  \lim_{z\to \eta_i} -\frac{N}{n}\frac{z-\eta_i}{z\hum(z)} = -\frac{N}{n} \frac1{\eta_i\hum'(\eta_i)}
\end{equation}
As for (2), it is the derivative of $-1/\hum(z)$, which is well-behaved inside $\uuContour$, so it does not have poles. The term (3) has poles of order $1$ in $z\in\{\eta_1,\ldots,\eta_N\}$ as well and we have the residue
\begin{equation}
  \lim_{z\to \eta_i} -\frac{N}{n}\frac{(z-\eta_i)\hm'(z)}{\hum(z)\hm(z)} = -\frac{N}{n}\frac{N}{M-N} \frac{\hm'(\eta_k)}{\hum'(\eta_k)}\eta_k
\end{equation}
the last equality being obtained from the fact that
\begin{equation}
  \hm(z) = \frac{M}{N}\hum(z)+\frac{M-N}{N}\frac1z
\end{equation}
and $\hm(\eta_i)=0$. It also has poles in $z\in\{\mu_1,\ldots,\mu_N\}$, where $\hm(z)=0$. These are order $1$ poles and we have
\begin{equation}
  \lim_{z\to \mu_i} -\frac{N}{n}\frac{(z-\mu_i)\hm'(z)}{\hum(z)\hm(z)} = \frac{N}{n}\frac{M}{M-N}\mu_k
\end{equation}
Term (4) is shown in a similar way to have residues $-\frac{N-n}n\frac{N}{M-N}\frac1{\eta_i}{\hum'(\eta_i)}$ and $\frac{N-n}n\frac{M}{M-N}\frac1{\mu_i}{\hm'(\mu_i)}$ , $i\in\{1,\ldots,N\}$. Term (5) has residues $\frac{N-n}n\frac{NM}{(M-N)^2}\eta_i$ and $-\frac{N-n}n\frac{M^2}{(M-N)^2}\frac{\mu_i\hum'(\mu_i)}{\hm'(\mu_k)}$, $i\in\{1,\ldots,N\}$. Term (6) has residues $-\frac{N-n}n\frac{N^2}{(M-N)^2}\frac{\eta_i\hm'(\eta_i)}{\hum'(\eta_i)}$ and $\frac{N-n}n\frac{MN}{(M-N)^2}\mu_i$, $i\in\{1,\ldots,N\}$. Term (7) has a pole in $z=0$ but we already know that $0$ is not inside $\uuContour$, so this is already discarded. Term (8) has poles in $z\in\{\lambda_1,\ldots,\lambda_N\}$ of residue $\frac{N\sigma^2}n$ and poles in $z\in\{\eta_1,\ldots,\eta_N\}$ of residue $-\frac{N\sigma^2}n$. Similarly term (9) has poles in $z\in\{\lambda_1,\ldots,\lambda_N\}$ of residue $\frac{N\sigma^2}n$ and poles in $z\in\{\mu_1,\ldots,\mu_N\}$ of residue $-\frac{N\sigma^2}n$.

Summing together the 9 terms and remarking that 
\begin{equation}
  \frac{N}{M-N}\hm'(z)=\frac{M}{M-N}\hum'(z)-\frac1{z^2}
\end{equation}
we obtain exactly \eqref{eq:residue_calc}.

Assume now $M=N$, in which case $\hum(z)=\hm(z)$. It can be readily seen that the terms (4) to (6) are the derivative of $-\frac{N-n}n\frac1{z^2\hm(z)}$, so that they have residue $0$. The only remaining term here is (1), whose residues are the $-\frac{N}n \frac1{\eta_i \hm'(\eta_i)}$.

\section{Proof of exact separation}
\label{app:bai}
Theorem \ref{th:bai} is a generalization from the assumption of identical distribution of the $x_{ij}$'s, the proof of which is contained in the two papers \cite{SIL98} and \cite{BAI99}, which, with some modifications, appear as Chapter 6 in \cite{SIL06}. The proof uses previous articles that need to be updated as well. We shall therefore go through the necessary steps that need to be modified, taking reference to all papers successively. 

We shall assume for simplicity in the following that the matrices $\T_n$ are deterministic, converges in distribution to $H$ and that $\Vert \T_n\Vert$ is uniformly bounded. The generalization to random $\T_n$ follows from Tonelli's theorem \cite{BIL08}. Indeed, let $\mathcal X$ be the probability space that engenders the $\X_n$ and $\mathcal T$ the probability space that engenders the $\T_n$. Let $A$ be any of the events in Conclusion 1), 2), or 3), claimed to occur with probability one. Assume that Theorem \ref{th:bai} holds for deterministic $\T_n$ satisfying Assumptions 5), 6), and 7). Let $t\in\mathcal T$ be an element of the intersection of these events. Then $I_A(t,x)=1$ for all $x$ contained in a subset of $\mathcal X$ having probability one. Therefore, by Tonelli's theorem, denoting $\mathcal T\times \mathcal X$ the product space of $\mathcal T$ and $\mathcal X$, we have that
\begin{align}
  & \int_{\mathcal T\times \mathcal X} I_A(t,x) dP_{\mathcal T\times \mathcal X}(t,x) \nonumber \\ 
  &= \int_{\mathcal T} \left[\int_{\mathcal X} I_A(t,x) dP_{\mathcal X}(x)\right]dP_{\mathcal T}(t) = 1,
\end{align}
and Theorem \ref{th:bai} therefore holds true if it holds true for $\T_n$ deterministic.

\subsection{Extension of \cite{YIN88}}
\label{app:YIN88}
The first step is to extend the work in \cite{YIN88} on the largest eigenvalue of $\S_n=\frac1n\X_n\X_n^\herm$, where $\X_n=(x_{ij})$ is $p\times n$, $p=p(n)$,
and $p/n\to y>0$ as $n\to\infty$.  Checking the assumptions in \cite{YIN88}, we change the six conditions of Page 518 to 
\begin{enumerate}
\item[(1)] $x_{ij}, i= 1, 2,\ldots, p;\ j= 1, 2,\ldots , n$ are independent for each $n$, 
\item[(2)] $|x_{ij}|<\eta_n\sqrt{n}$, where $\eta_n\downarrow 0$, 
\item[(3)] $\EE x_{ij} = 0$, 
\item[(4)] $\EE |x_{ij}^ 2| \le 1$,
\item[(5)] $\EE|x_{ij}|^l<(\eta_n\sqrt{n})^{l-1}$, for $l\ge 2$,
\item[(6)] $\EE|x_{ij}^l|\le c(\eta_n\sqrt{n})^{l-3}$, for $l\ge 3$.
\end{enumerate}

By the same argument given there (no difference for complex random variables), the inequality
\begin{equation}
\EE\tr(\S_n)^k\le \eta^k
\end{equation}
holds for any $\eta>b\equiv(1+\sqrt{y})^2$ provided $k$ is chosen such that 
\begin{enumerate}
\item[(a)] $k/\log n\to\infty$, 
\item[(b)]$\eta_n^{\frac16} k/\log n\to 0$.
\end{enumerate}

This implies that $\P(\lambda_{\max}(\S_n)>b+\varepsilon)=o(n^{-t})$ for any given $\varepsilon>0$ and $t>0$.

\begin{remark}
Notice that if Condition (4) is replaced by $\EE|x_{ij}^2|\le \iota$, where $\iota$ is a fixed positive constant, then we have 
\begin{equation}
\P(\lambda_{\max}(\S_n)>\iota (b+\varepsilon))=o(n^{-t}). \label{eqybk}
\end{equation}

We only need to consider the matrix $\iota^{-1}\S_n$ and replace $x_{ij}$ by $\iota^{-1/2}x_{ij}$ to verify the six conditions.
\end{remark}

\subsection{First step truncation and renormalization}
We consider $\S_n\T_n$, where the assumptions of Theorem \ref{th:bai} are met, except the $\T_n$ are assumed nonrandom.  Here, to be consistent with Chapter 6 of \cite{SIL06}, we replace $c_n$ with $y_n$ and $c$ with $y$.

Notice, from the identity 
\begin{equation} 
  \EE Y^4=\int_0^{\infty} P(Y>x^{1/4})dx,
\end{equation} 
valid for any nonnegative random variable $Y$, that (ii) implies the fourth moments of the $x_{ij}$ exist.

We will need the following identity later on. For nonnegative $Y$ having finite fourth moment, since $\EE Y^4 I(Y>y)\geq y^4 P(Y>y)$, we have $y^4 P(Y>y)\to 0$ as $y\to\infty$. Thus, using integration by parts, we have for any $a>0$
\begin{align}
  &\EE Y^4I(Y>a) \nonumber \\ 
  & =a^4 P(Y>a)+\lim_{y\to\infty}(-y^4 P(Y>y)+\int_a^y4x^{3}P(Y>x)dx) \\
  & =a^4P(Y>a)+\int_a^{\infty}4x^{3}P(Y>x)dx.
\end{align}

We choose $\eta_n\downarrow 0$ such that $\eta_n\sqrt{n}\uparrow \infty$, $\liminf_n\eta_n^2\sqrt n>0$, and
\begin{equation}
\sum_{k=1}^\infty 2^{2k}P(|X|>\tilde\eta_k 2^{k/2})<\infty,
\end{equation}
where $\tilde\eta_k=\eta_{2^{k}}$.

{\bf 1. Truncation.} Define $y_{ij}=x_{ij}I(|x_{ij}|\leq \eta\sqrt{n})$, $\Y_n=(y_{ij})_{p\times n}$ and $\widehat\S_n=\frac1n\Y_n\Y_n^\herm$. We have
\begin{align}
& P(\X_n\ne\Y_n,\textmd{ i.o.}) \nonumber \\
&\leq
\lim_{m\to\infty}\sum_{k=m}^\infty P\left(\bigcup_{n=2^k+1}^{2^{k+1}} \bigcup_{i\le p,j\le n}\{|x_{ij}|>\eta_n\sqrt{n}\}\right)\\
&\le
\lim_{m\to\infty}\sum_{k=m}^\infty P\left(\bigcup_{n=2^k+1}^{2^{k+1}}
\bigcup_{i\le 2yn,j\le n}\{|x_{ij}|>\tilde \eta_{{k}}2^{k/2}\}\right)\\
&=\lim_{m\to\infty}\sum_{k=m}^\infty P\left(
\bigcup_{i\le y2^{k+2},j\le 2^{k+1}}\{|x_{ij}|>\tilde\eta_{{k}}2^{k/2}\}\right)\\
&= 8yK\lim_{m\to\infty}\sum_{k=m}^\infty 2^{2k} P\left( |X|>\tilde
\eta_{{k}}2^{k/2}\right)=0.
\end{align}

{\bf 2. Centralization.}
Define $z_{ij}=y_{ij}-\EE y_{ij}$ and $\Z_n=(z_{ij})_{p\times n}$ and $\widetilde \S_n=\frac1n\Z_n\Z_n^\herm$. Then by Theorem A. 46 of \cite{SIL06} and the above identity, we have 
\begin{align}
  &\max_{k\le p}|\lambda_k^\oh(\widehat\S_n\T_n)-\lambda_k^\oh(\widetilde\S_n\T_n)| \nonumber \\ 
  &\le \|\T_n^\oh\|\|n^{-\oh}\EE(\Y_n)\|\\
  &\le \left(\frac1n\sum_{ij}|\EE y_{ij}|^2\right)^\oh\\
  &= \left(\frac1n\sum_{ij}|\EE x_{ij}I(|x_{ij}|>\eta_n\sqrt{n})|^2\right)^\oh\\
  &\le\left(\frac1{\eta_n^4n^3}\sum_{ij}\EE |x_{ij}^2|\EE|x_{ij}^4|I(|x_{ij}|>\eta_n\sqrt{n})\right)^\oh\\
  &\le \left(\frac{Knp}{\eta_n^4n^3} \EE|X^4|I(|X|>\eta_n\sqrt{n})\right)^\oh\to 0.
\end{align}

{\bf 3. Rescaling.}
Define $w_{ij}=z_{ij}/\sigma_{ij}$, $\W_n=(w_{ij})_{p\times n}$ and $\breve\S_n=\frac1n\W\W_n^\herm$. Then, by Theorem A. 46 of \cite{SIL06}, 
\begin{align}
&\max_{k\le p}|\lambda_k^\oh(\breve\S_n\T_n)-\lambda_k^\oh(\widetilde\S_n\T_n)| \nonumber \\
& \le  \|\T_n^\oh\|\|n^{-\oh}(\Z_n-\W_n)\|\\
&\le   \left\|\frac1{\sqrt{n}}\left[z_{ij}(1-\lambda_{ij}^{-1})\right]\right\|\to 0\textmd{, a.s.}
\end{align}
because of \eqref{eqybk}  and the fact that
\begin{align}
&\max_{i,j}|1-\lambda_{ij}^2| \\
&\le \max_{ij}[\EE|x_{ij}^2|I(|x_{ij}|>\eta\sqrt{n})+(\EE|x_{ij}|I(|x_{ij}|>\eta\sqrt{n}))^2]\\
&\le 2\psi^{-1}(\eta\sqrt{n})\max_{ij}\EE|x_{ij}^2|\psi(|x_{ij}|)\to 0 
\end{align}
which implies that
\begin{equation}
\max_{ij}\EE|w_{ij}-z_{ij}|^2=\max_{ij}\frac{(1-\lambda_{ij}^2)^2}{(1+\lambda_{ij})^2}\to 0.
\end{equation}

\subsection{Second truncation and normalization}
We may assume now that the $x_{ij}$ satisfy the six conditions of Section \ref{app:YIN88} with Condition (4) strengthened to $\EE|x_{ij}^2|=1$ for all $i,j$.

Define $y_{ij}=x_{ij}I(|x_{ij}|\le C)-\EE x_{ij}I(|x_{ij}|\le C)$ for some large constant $C$ and define $\Y_n=(y_{ij})_{p\times n}$, $\widehat\S_n=\frac1n\Y_n\Y_n^\herm$.

Then by Theorem A. 46 of \cite{SIL06}, we have 
\begin{align}
&\max_{k\le p}|\lambda_k^\oh(\widehat\S_n\T_n)-\lambda_k^\oh(\widetilde\S_n\T_n)| \nonumber \\ 
&\le \|\T_n^\oh\|\left\|n^{-1/2}(\X_n-\Y_n)\right\|\\
&\le \left\|n^{-1/2}(\X_n-\Y_n)\right\|.
\end{align}

Since $\EE|x_{ij}-y_{ij}|^2\le \EE|x_{ij}^2|I(|x_{ij}|>C)\le M/\psi(C)$, this can be made arbitrarily small by making $C$ sufficiently large. We can then apply \eqref{eqybk}.

The rescaling is the same as given in last Section. Now we have
\begin{equation}
  \max_{ij}|1-\sigma_{ij}^2|\leq2\psi^{-1}(C)\max_{ij}\EE|x^2_{ij}|\psi(x_{ij}),
\end{equation}
which can be made arbitrarily small by making $C$ sufficiently large.

\subsection{Extension of \cite{BAI93} and Chapter 6 of \cite{SIL06}}
The result in \cite{BAI93} on the smallest eigenvalue, $\lambda_{\min}(\S_n)$, of $\S_n$ can be extended with only Assumptions 1), 2), 3) of Theorem \ref{th:bai}. Indeed, using the two step truncations, we may assume the $x_{ij}$ are bounded, with mean $0$ and variance $1$. Following the same steps as in \cite{BAI93}, one may prove that when $y<1$
\begin{equation}
\lambda_{\min}(\S_n)\to(1-\sqrt y)^2,
\end{equation}
almost surely.

We proceed now to the necessary changes in Chapter 6 in \cite{SIL06}. We may now assume the same conditions as in Section 6.2.1 of \cite{SIL06} on the $x_{ij}$ (except they need not be identically distributed), and the bounds appearing there. The changes are needed wherever identical distribution was exploited.

We begin with Page 139, below (6.2.34). We change the definition of $b_n$ to 
\begin{equation}
b_n=\frac1{1+n^{-1}\EE\tr(\T_n\D^{-1})},
\end{equation}
and introduce the quantities
\begin{equation}
b_{nj}=\frac1{1+n^{-1}\EE\tr(\T_n\D^{-1}_j)}.
\end{equation}

The argument below (6.2.35) is specific for $j=1$, but easily extends for any $j$. Following the argument below (6.2.36), we can no longer assume $\EE\beta_1=-z\EE\underline s_n$, nor is bounded, but we have
\begin{equation}
\sup_{u\in[a,b]}\left|\frac1n\sum_k\EE\beta_k\right|\leq K.
\end{equation}

We further have 
\begin{equation}
  \label{eq:jackupdate1}
b_{nk}=\beta_k+\beta_kb_{nk}\gamma_k.
\end{equation}

Then, using (6.2.36)
 \begin{align}
   \frac1n\left|\sum_k(b_{nk}-\EE\beta_k)\right|&\leq Kn^{-1}\sum_kv_n^{-2}
(\EE|\gamma_k|^2)^\oh \\ &\leq v_n^{-3}n^{-\oh}.
\end{align}

Since $b_{nj}-b_n=b_nb_{nj}\EE(\frac1n\tr\T(\D^{-1}-\D_j^{-1}))$, we have, using Lemma 6.9 of \cite{SIL06}
\begin{equation}
\left|b_n-\frac1n\sum_kb_{nk}\right|=\frac1n\left|\sum_k(b_n-b_{nk})\right|\leq |z|^2\frac1{nv^3},
\end{equation}
and
\begin{equation}
|b_{nk}-b_n|\leq K\frac1{nv^3}.
\end{equation}

Thus we have 
\begin{equation}
\max_j\sup_{u\in[a,b]}|b_{nj}|\leq K.
\end{equation}

For the rest of Section 6.2.3, $b_n$ is mentioned twice. We need to replace it with $b_{nj}$ and the arguments go through without any further changes.

For Section 6.2.4, (6.2.42) needs to be replaced by 
\begin{align}
& y_n\int\frac{dH_n(t)}{1+t\EE\underline s_n}+zy_n\EE(s_n(z)) \\ 
&=\frac1n\sum_{k=1}^n\EE\beta_k\left[\r^\herm_k\D_k^{-1}(\EE\underline s_n\T_n+\I)^{-1}\r_k \right. \nonumber \\ & \left. -\frac1n\EE\tr(\EE\underline s_n\T_n+\I)^{-1}\T_n\D^{-1}\right].
\end{align}

For the rest of the section, replace subscript $1$ with subscript $k$, subscript $2$ with subscript $j$, replace $b_{1n}$ with 
\begin{equation}
b_{kj}=\frac1{1+n^{-1}\EE\tr(\T_n\D_{kj})^{-1}},~ k\neq j,
\end{equation}
and all appearances of subscripts $kj$ assume $k\neq j$. $F_{nkj}$ has the obvious definition. Replace the summations for $j$ ranging from $2$ to $n$ with $j\neq k$. All the bounds derived for $k=1$ are true for all $k$.  So we conclude the left side of (6.2.42) is bounded by $kn^{-1}$.

The rest of Chapter 6 follows without any changes.

\bibliography{tutorial_RMT/IEEEconf,tutorial_RMT/IEEEabrv,tutorial_RMT/tutorial_RMT}

\begin{thebibliography}{10}
\providecommand{\url}[1]{#1}
\csname url@samestyle\endcsname
\providecommand{\newblock}{\relax}
\providecommand{\bibinfo}[2]{#2}
\providecommand{\BIBentrySTDinterwordspacing}{\spaceskip=0pt\relax}
\providecommand{\BIBentryALTinterwordstretchfactor}{4}
\providecommand{\BIBentryALTinterwordspacing}{\spaceskip=\fontdimen2\font plus
\BIBentryALTinterwordstretchfactor\fontdimen3\font minus
  \fontdimen4\font\relax}
\providecommand{\BIBforeignlanguage}[2]{{%
\expandafter\ifx\csname l@#1\endcsname\relax
\typeout{** WARNING: IEEEtran.bst: No hyphenation pattern has been}%
\typeout{** loaded for the language `#1'. Using the pattern for}%
\typeout{** the default language instead.}%
\else
\language=\csname l@#1\endcsname
\fi
#2}}
\providecommand{\BIBdecl}{\relax}
\BIBdecl

\bibitem{MIT99}
J.~M. III and G.~Q.~M. Jr, ``{Cognitive radio: making software radios more
  personal},'' \emph{{IEEE} Personal Commun. Mag.}, vol.~6, no.~4, pp. 13--18,
  1999.

\bibitem{CLA08}
H.~Claussen, L.~T. Ho, and L.~G. Samuel, ``{An overview of the femtocell
  concept},'' \emph{Bell Labs Technical Journal}, vol.~13, no.~1, pp. 221--245,
  May 2008.

\bibitem{CAL10}
D.~Calin, H.~Claussen, and H.~Uzunalioglu, ``{On femto deployment architectures
  and macrocell offloading benefits in joint macro-femto deployments},''
  \emph{{IEEE} Trans. Commun.}, vol.~48, no.~1, pp. 26--32, Jan. 2010.

\bibitem{KON09}
V.~Chandrasekhar, M.~Kountouris, and J.~G. Andrews, ``{Coverage in
  Multi-Antenna Two-Tier Networks},'' \emph{{IEEE} Trans. Wireless Commun.},
  vol.~8, no.~10, pp. 5314--5327, 2009.

\bibitem{URK67}
H.~Urkowitz, ``{Energy detection of unknown deterministic signals},''
  \emph{Proc. {IEEE}}, vol.~55, no.~4, pp. 523--531, 1967.

\bibitem{KOS02}
V.~I. Kostylev, ``{Energy detection of a signal with Random Amplitude},'' in
  \emph{Proc. {IEEE} International Conference on Communications (ICC'02)}, New
  York, NY, USA, 2002, pp. 1606--1610.

\bibitem{COU09b}
R.~Couillet and M.~Debbah, ``{A Bayesian framework for collaborative
  multi-source signal detection},'' \emph{{IEEE} Trans. Signal Process.},
  vol.~58, no.~10, pp. 5186--5195, Oct. 2010.

\bibitem{BIA10}
P.~Bianchi, J.~Najim, M.~Maida, and M.~Debbah, ``{Performance of Some
  Eigen-based Hypothesis Tests for Collaborative Sensing},'' \emph{{IEEE}
  Trans. Inf. Theory}, 2010, to appear.

\bibitem{HER07}
P.~Chung, J.~B\"ohme, C.~Mecklenbra\"uker, and A.~Hero, ``{Detection of the
  Number of Signals Using the Benjamini-Hochberg Procedure},'' \emph{{IEEE}
  Trans. Signal Process.}, vol.~55, no.~6, pp. 2497--2508, 2007.

\bibitem{SIL95}
J.~W. Silverstein and Z.~D. Bai, ``{On the empirical distribution of
  eigenvalues of a class of large dimensional random matrices},'' \emph{Journal
  of Multivariate Analysis}, vol.~54, no.~2, pp. 175--192, 1995.

\bibitem{SIL92}
J.~W. Silverstein and P.~L. Combettes, ``{Signal detection via spectral theory
  of large dimensional random matrices},'' \emph{{IEEE} Trans. Signal
  Process.}, vol.~40, no.~8, pp. 2100--2105, 1992.

\bibitem{KAR08}
N.~E. Karoui, ``{Spectrum estimation for large dimensional covariance matrices
  using random matrix theory},'' \emph{Annals of Statistics}, vol.~36, no.~6,
  pp. 2757--2790, Dec. 2008.

\bibitem{RAO08}
N.~R. Rao, J.~A. Mingo, R.~Speicher, and A.~Edelman, ``{Statistical
  eigen-inference from large Wishart matrices},'' \emph{Annals of Statistics},
  vol.~36, no.~6, pp. 2850--2885, Dec. 2008.

\bibitem{COU08}
R.~Couillet and M.~Debbah, ``{Free deconvolution for OFDM multicell SNR
  detection},'' in \emph{Proc. {IEEE} International Symposium on Personal,
  Indoor and Mobile Radio Communications (PIMRC'08)}, Cannes, France, 2008.

\bibitem{MES08}
X.~Mestre, ``{On the asymptotic behavior of the sample estimates of eigenvalues
  and eigenvectors of covariance matrices},'' \emph{{IEEE} Trans. Signal
  Process.}, vol.~56, no.~11, pp. 5353--5368, Nov. 2008.

\bibitem{SIL06}
Z.~Bai and J.~W. Silverstein, ``{Spectral Analysis of Large Dimensional Random
  Matrices},'' \emph{Springer Series in Statistics}, 2009.

\bibitem{BAI99}
Z.~D. Bai and J.~W. Silverstein, ``{Exact Separation of Eigenvalues of Large
  Dimensional Sample Covariance Matrices},'' \emph{The Annals of Probability},
  vol.~27, no.~3, pp. 1536--1555, 1999.

\bibitem{RUD86}
W.~Rudin, \emph{{Real and complex analysis}}, 3rd~ed.\hskip 1em plus 0.5em
  minus 0.4em\relax McGraw-Hill Series in Higher Mathematics, May 1986.

\bibitem{CHO95}
J.~W. Silverstein and S.~Choi, ``{Analysis of the limiting spectral
  distribution of large dimensional random matrices},'' \emph{Journal of
  Multivariate Analysis}, vol.~54, no.~2, pp. 295--309, 1995.

\bibitem{BIL08}
P.~Billingsley, \emph{{Probability and Measure}}, 3rd~ed.\hskip 1em plus 0.5em
  minus 0.4em\relax Hoboken, NJ: John Wiley \& Sons, Inc., 1995.

\bibitem{SIL98}
Z.~D. Bai and J.~W. Silverstein, ``{No Eigenvalues Outside the Support of the
  Limiting Spectral Distribution of Large Dimensional Sample Covariance
  Matrices},'' \emph{Annals of Probability}, vol.~26, no.~1, pp. 316--345, Jan.
  1998.

\bibitem{GRE09}
D.~Gregoratti and X.~Mestre, ``{Random DS/CDMA for the amplify and forward
  relay channel},'' \emph{{IEEE} Trans. Wireless Commun.}, vol.~8, no.~2, pp.
  1017--1027, 2009.

\bibitem{RYA07}
O.~. Ryan and M.~Debbah, ``{Free deconvolution for signal processing
  applications},'' in \emph{Proc. {IEEE} International Symposium on Information
  Theory (ISIT'07)}, Nice, France, Jun. 2007, pp. 1846--1850.

\bibitem{SER00}
R.~S\'eroul, \emph{{Programming for Mathematicians}}.\hskip 1em plus 0.5em
  minus 0.4em\relax New York, NY, USA: Springer Universitext, Feb. 2000.

\bibitem{MES08c}
X.~Mestre and M.~Lagunas, ``{Modified Subspace Algorithms for DoA Estimation
  With Large Arrays},'' \emph{{IEEE} Trans. Signal Process.}, vol.~56, no.~2,
  pp. 598--614, Feb. 2008.

\bibitem{YIN88}
Y.~Q. Yin, Z.~D. Bai, and P.~R. Krishnaiah, ``{On the limit of the largest
  eigenvalue of the large dimensional sample covariance matrix},''
  \emph{Probability Theory and Related Fields}, vol.~78, no.~4, pp. 509--521,
  1988.

\bibitem{BAI93}
Z.~D. Bai and Y.~Q. Yin, ``{Limit of the smallest eigenvalue of a large
  dimensional sample covariance matrix},'' \emph{The Annals of Probability},
  vol.~21, no.~3, pp. 1275--1294, 1993.

\end{thebibliography}

\begin{IEEEbiography}[{\includegraphics[width=1in,height=1.25in]{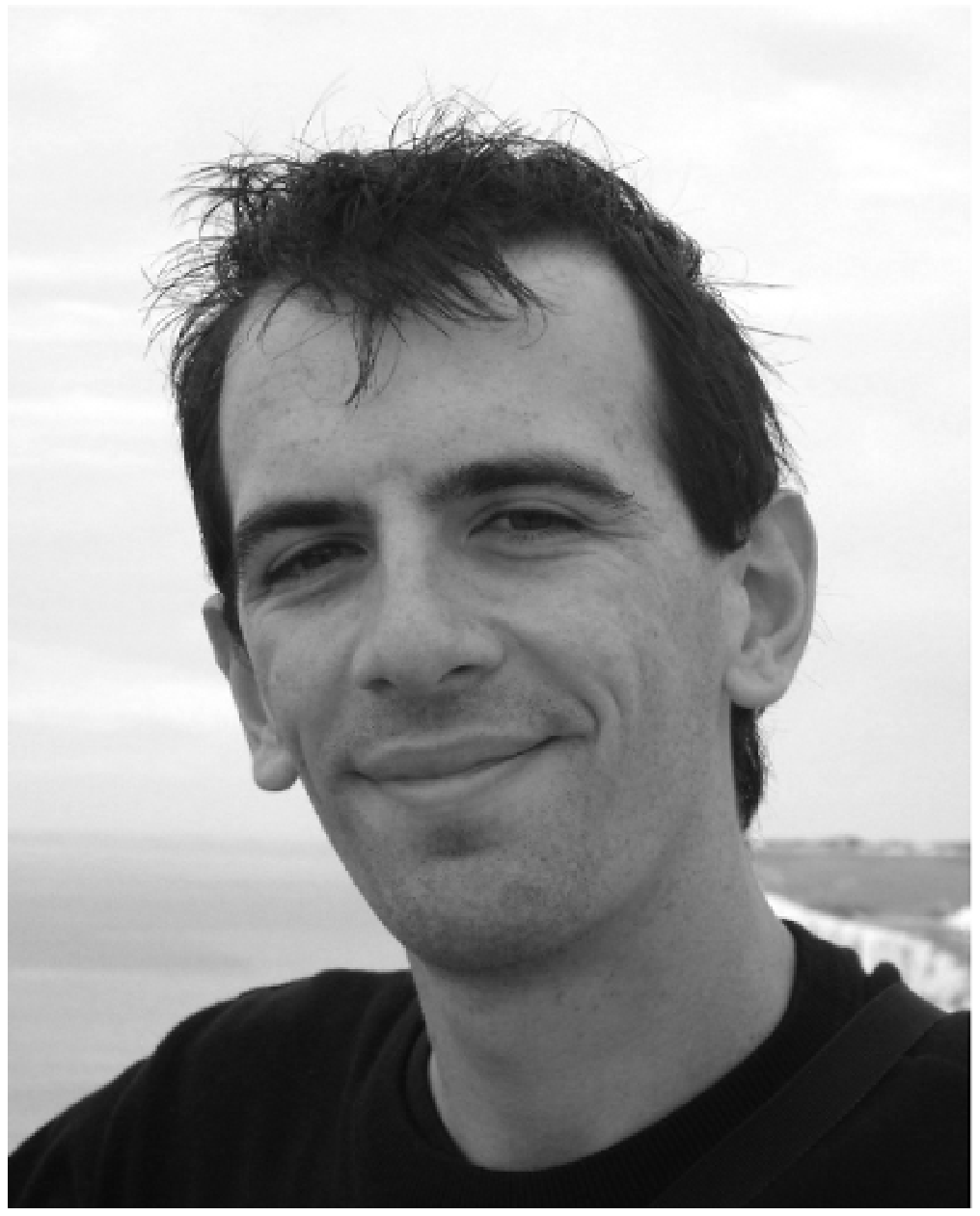}}]{Romain Couillet}
was born in Abbeville, France. He received his Msc. in Mobile Communications at the Eurecom Institute, France in 2007. He received his Msc. in Communication Systems in Telecom ParisTech, France in 2007. In September 2007, he joined ST-Ericsson (formerly NXP Semiconductors, founded by Philips). At ST-Ericsson, he works as an Algorithm Development Engineer on the Long Term Evolution Advanced (LTE-A) project. In parallel to his position at ST-Ericsson, he is currently a PhD student at Supélec, France. His research topics include mobile communications, multi-users multi-antenna detection, cognitive radio cognitive, Bayesian probability and random matrix theory. He is the recipient of the ValueTools best student paper award, 2008.
\end{IEEEbiography}

\begin{IEEEbiography}[{\includegraphics[width=1in,height=1.25in]{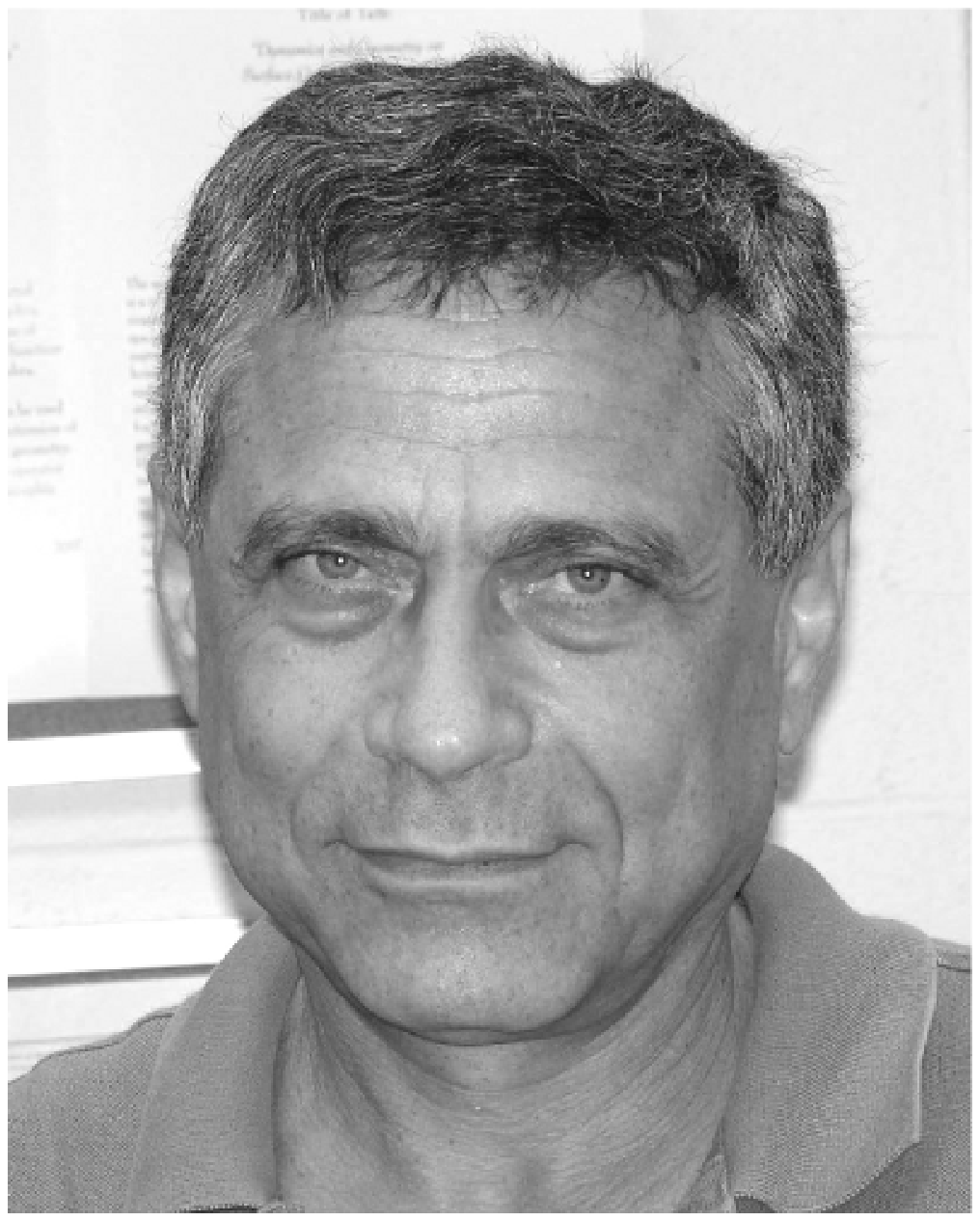}}]{Jack W. silverstein}
received the B.A. degree in mathematics from Hofstra University, Hempstead, NY, in 1971 and the M.S. and Ph.D. degrees in applied mathematics from Brown University, Providence, RI, in 1973 and 1975, respectively.  After postdoctoring and teaching at Brown, he began in 1978 a tenure track position in the Department of Mathematics at North Carolina State University, Raleigh, where he has been Professor since 1994. His research interests are in probability theory with emphasis on the spectral behavior of large-dimensional random matrices. Prof. Silverstein was elected Fellow of the Institute of Mathematical Statistics in 2007. 
\end{IEEEbiography}

\begin{IEEEbiography}[{\includegraphics[width=1in,height=1.25in]{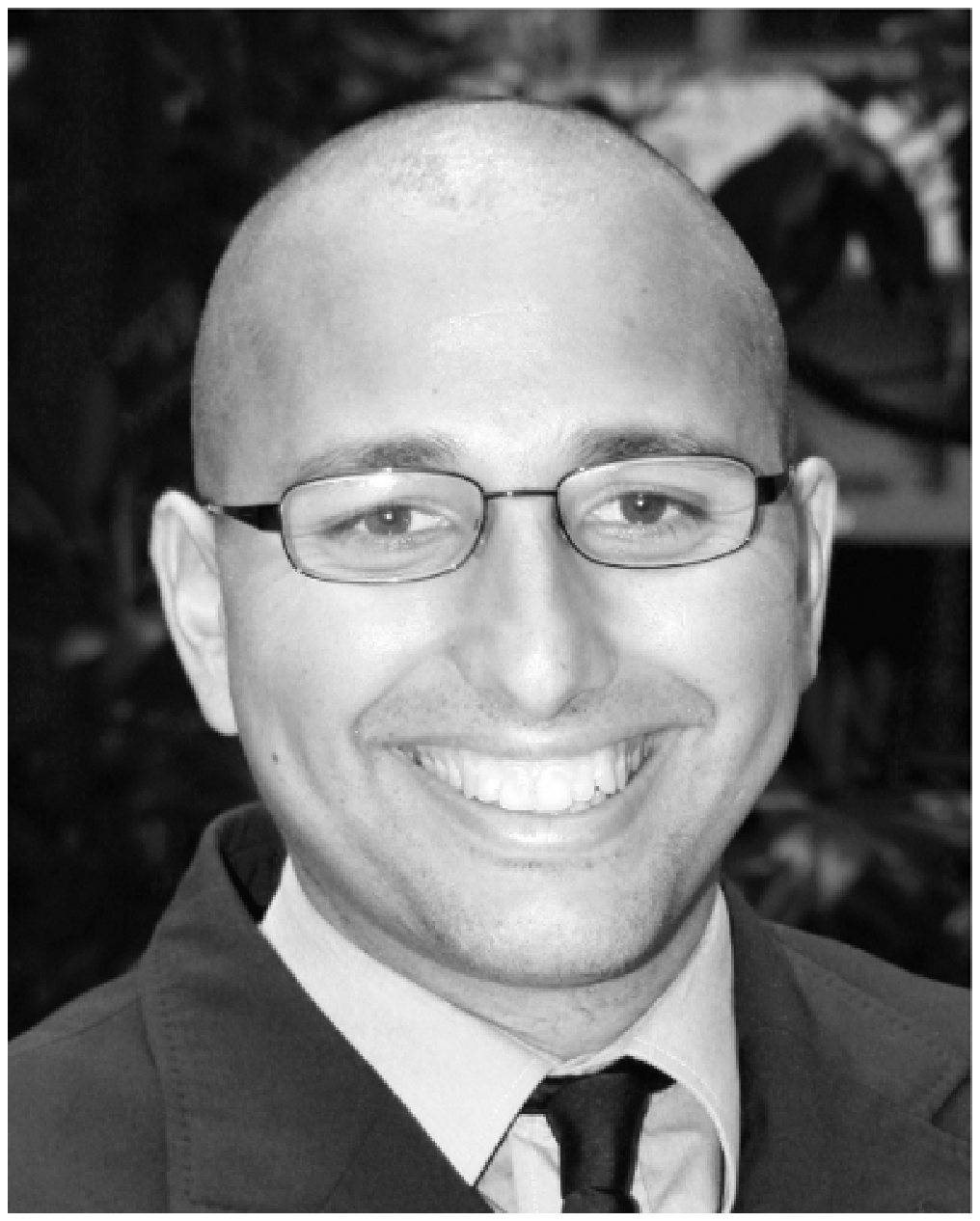}}]{M\'erouane Debbah}
was born in Madrid, Spain. He entered the Ecole Normale Supérieure de Cachan (France) in 1996 where he received his M.Sc and Ph.D. degrees respectively in 1999 and 2002. From 1999 to 2002, he worked for Motorola Labs on Wireless Local Area Networks and prospective fourth generation systems. From 2002 until 2003, he was appointed Senior Researcher at the Vienna Research Center for Telecommunications (FTW) (Vienna, Austria). From 2003 until 2007, he joined the Mobile Communications de-partment of the Institut Eurecom (Sophia Antipolis, France) as an Assistant Professor. He is presently a Professor at Supelec (Gif-sur-Yvette, France), holder of the Alcatel-Lucent Chair on Flexible Radio. His research interests are in information theory, signal processing and wireless communications. Mérouane Debbah is the recipient of the ``Mario Boella'' prize award in 2005, the 2007 General Symposium IEEE GLOBECOM best paper award, the Wi-Opt 2009 best paper award, the2010  Newcom++ best paper award  as well as the Valuetools 2007,Valuetools 2008 and CrownCom2009 best student paper awards. He is a WWRF fellow. 
\end{IEEEbiography}

\begin{IEEEbiography}[{\includegraphics[width=1in,height=1.25in]{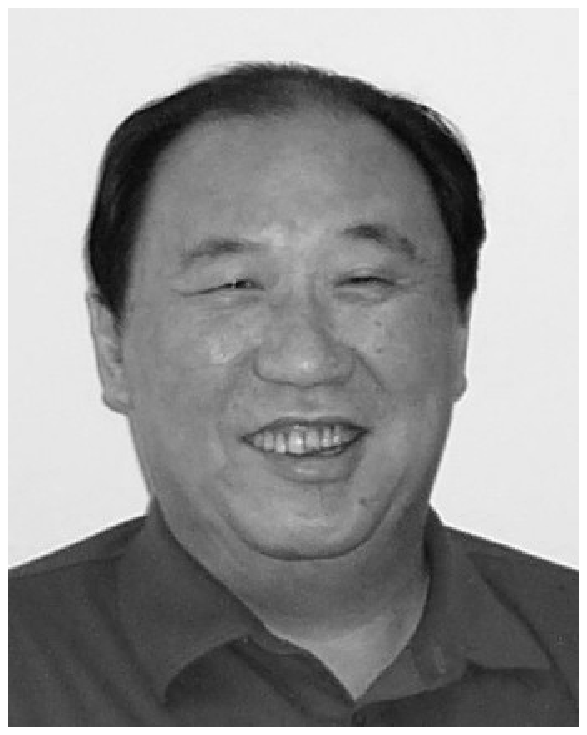}}]{Zhidong Bai}
graduated from University of Science and Technology of China, majoring statistics and probability. His research interests include limiting theorems of statistics and spectral analysis of large dimensional random matrices, rounded data analysis etc. He is currently a professor of the School of Mathematics and Statistics at Northeast Normal University, China, and Department of Statistics and Applied Probability at National University of Singapore. He is a Fellow of the Third World Academy of Sciences and a Fellow of the Institute of Mathematical Statistics.
\end{IEEEbiography}

\end{document}